\documentclass[a4paper,10pt]{scrartcl}
\pdfoutput=1 

\usepackage[utf8]{inputenc}
\usepackage[USenglish]{babel}
\usepackage{pifont}
\usepackage{graphicx}
\usepackage{xspace}
\usepackage{bbm}
\usepackage{mathtools}
\usepackage{nicefrac}
\usepackage{stmaryrd}
\SetSymbolFont{stmry}{bold}{U}{stmry}{m}{n}
\usepackage{anyfontsize}
\usepackage{booktabs}
\usepackage{colortbl}
\DeclareFontShape{OT1}{cmr}{bx}{sc}{<-> cmbcsc10}{}

\usepackage[dvipsnames]{xcolor}
\usepackage{amssymb,amsmath,amsthm,amsfonts}
\usepackage{authblk}

\usepackage{fullpage}
\usepackage[labelfont={sf,bf}]{caption}

\definecolor{darkgray}{rgb}{0.33, 0.33, 0.33}
\definecolor{lightgray}{rgb}{0.6, 0.6, 0.6}
\definecolor{myred}{RGB}{255,15,0}
\definecolor{ao(english)}{rgb}{0.0, 0.5, 0.0} 

\definecolor{itemizecol}{named}{darkgray}
\definecolor{enumcol}{named}{darkgray}
\definecolor{desccol}{named}{darkgray}

\usepackage{hyperref}
\hypersetup{pdfauthor={S. Winter, M. Zimmermann},pdftitle={Weak Muller Conditions Make Delay Games Hard},colorlinks=true,linkcolor={myred},citecolor={ao(english)}}
\usepackage[nameinlink,capitalise]{cleveref}
\crefname{prop}{Proposition}{Propositions}
\crefname{lem}{Lemma}{Lemmas}
\crefname{ex}{Example}{Examples}
\crefname{item}{Item}{Items}

\usepackage{tikz}
\usetikzlibrary{arrows,automata,trees,backgrounds,decorations.pathmorphing,decorations.pathreplacing,positioning,calc,shapes.geometric}
\tikzset{shorten >=1pt, >=stealth, auto, node distance=6em, initial text=}

\makeatletter
\def\th@plain{%
  \thm@headfont{\bfseries\sffamily}
  \thm@notefont{\normalfont\sffamily}
  \itshape 
}
\def\th@definition{%
  \thm@headfont{\bfseries\sffamily}
  \thm@notefont{\normalfont\sffamily}
  \normalfont 
}
\makeatother

\theoremstyle{plain}
\newtheorem{theorem}{Theorem}[section] 
\newtheorem{lemma}[theorem]{Lemma}
\newtheorem{corollary}[theorem]{Corollary}
\newtheorem{proposition}[theorem]{Proposition}

\theoremstyle{definition}

\newtheorem{example}[theorem]{Example}

\theoremstyle{remark}

\newcommand{\myquot}[1]{``#1''}

\newcommand{\nats}{\mathbbm{N}}
\newcommand{\size}[1]{|#1|}

\renewcommand{\epsilon}{\varepsilon}
\renewcommand{\phi}{\varphi}

\newcommand{\set}[1]{\{#1\}}
\newcommand{\pow}[1]{2^{#1}}

\newcommand{\aut}{\mathfrak{A}}

\newcommand{\autp}{\mathfrak{P}}

\newcommand{\TM}{\mathcal M}

\newcommand{\acc}{\mathrm{Acc}}
\newcommand{\col}{\Omega}
\newcommand{\infi}[0]{\mathrm{Inf}}
\newcommand{\occ}[0]{\mathrm{Occ}}
\newcommand{\curlyF}[0]{\mathcal{F}}
\newcommand{\I}{\mathcal{I}}

\newcommand{\initmark}{I}

\newcommand{\formula}{\ensuremath{\mathsf{Emerson}}-\ensuremath{\mathsf{Lei}}}
\newcommand{\explicit}{\ensuremath{\mathsf{explicit}}}

\newcommand{\wDMA}{\ensuremath{\mathsf{wDMA}}\xspace}
\newcommand{\wNMA}{\ensuremath{\mathsf{wNMA}}\xspace}

\newcommand{\DPA}{\ensuremath{\mathsf{DPA}}\xspace}
\newcommand{\NPA}{\ensuremath{\mathsf{NPA}}\xspace}

\newcommand{\delaygame}[1]{\Gamma\!_{f}(#1)}
\newcommand{\delaygamep}[1]{\Gamma\!_{g}(#1)}

\newcommand{\SigmaI}{\Sigma_I}
\newcommand{\SigmaO}{\Sigma_O}

\newcommand{\stratO}{\tau_O}
\newcommand{\stratI}{\tau_I}
\newcommand{\p}{P}

\newcommand{\bigo}{\mathcal{O}}

\usepackage{xspace}
\newcommand{\ptime}{{\upshape{\textsc{P}}}\xspace}
\newcommand{\np}{{\upshape{\textsc{NP}}}\xspace}
\newcommand{\conp}{{\upshape{\textsc{co-NP}}}\xspace}

\newcommand{\pspace}{{\upshape{\textsc{PSpace}}}}
\newcommand{\exptime}{{\upshape{\textsc{ExpTime}}}\xspace}

\newcommand{\twoexp}{{\upshape{\textsc{2ExpTime}}}\xspace}
\newcommand{\threeexp}{{\upshape{\textsc{3ExpTime}}}\xspace}
\newcommand{\atwoexpspace}{\upshape{\textsc{A2ExpSpace}}\xspace}
\newcommand{\aexpspace}{\upshape{\textsc{AExpSpace}}\xspace}

\newcommand{\bool}{\mathbb{B}}

\newcommand{\nodelbl}[2]{ \ensuremath{{\scriptstyle b_{#1} = #2}} }

\newcommand{\zero}[2]{\mathit{S}(\nodelbl{#1}{#2})}
\newcommand{\one}[2]{\mathit{M}_1(\nodelbl{#1}{#2})}
\newcommand{\two}[2]{\mathit{M}_2(\nodelbl{#1}{#2})}

\newcommand{\res}[3]{ \ensuremath{{\scriptstyle \begin{smallmatrix}
  \mathrm{pos} = {#1}\\
  \mathrm{res} = #2\\
  \mathrm{carry} = #3
\end{smallmatrix}}}}
\newcommand{\add}[3]{\mathit{A}\left(\res{#1}{#2}{#3}\right)}

\newcommand{\addi}[2]{\mathit{A}(\nodelbl{#1}{#2})}
\newcommand{\addequal}[2]{\mathit{A}_{=}(\nodelbl{#1}{#2})}
\newcommand{\addguess}[2]{\mathit{G}(\nodelbl{#1}{#2})}

\newcommand{\yes}{\text{\ding{53}}}%
\newcommand{\no}{\text{\ding{107}}}%
\newcommand{\noo}{\text{\ding{54}}}%

\newcommand{\ok}{\text{\ding{51}}}%
\newcommand{\nok}{\text{\ding{55}}}%

\title{Weak Muller Conditions\\Make Delay Games Hard}

\author{Sarah Winter}  
\affil{Université libre de Bruxelles, Brussels, Belgium\\
\texttt{sarah.winter@ulb.ac.be}}

\author{Martin Zimmermann}
\affil{Aalborg University, Aalborg, Denmark\\
\texttt{mzi@cs.aau.dk}}

\date{}

\begin{document}

\maketitle






\begin{abstract}
We show that solving delay games with winning conditions given by deterministic and non-deterministic weak Muller automata is 2EXPTIME-complete respectively 3EXPTIME-complete. 
Furthermore, doubly- and triply-exponential lookahead is necessary and sufficient to win such games.
These results are the first that show that the succinctness of the automata types used to specify the winning conditions has an influence on the complexity of these problems.
\end{abstract}

\section{Introduction}
\label{sec:intro}

Is solving delay games with Muller conditions, i.e., determining its winner, harder than solving delay games with parity conditions? 
Is more lookahead required to win a delay game with a Muller condition than to win a delay game with a parity condition?
Deterministic Muller automata are exponentially more succinct than deterministic parity automata~\cite{DBLP:conf/csl/Boker17}\footnote{Strictly speaking this is only true when the size of automata is measured in the number of states, which is a very crude measure for Muller automata. See the discussion about the representation of Muller conditions later in the introduction and also the works of Boker~\cite{DBLP:conf/csl/Boker17} and Hunter and Dawar~\cite{DBLP:conf/mfcs/HunterD05,HunterPhD}.} and solving classical, delay-free Muller games is \pspace-complete~\cite{DBLP:conf/mfcs/HunterD05} while solving parity games is quasi-polynomial~\cite{DBLP:conf/stoc/CaludeJKL017}.
So, the answer to both questions should be yes, but surprisingly, and frustratingly, both questions are open. 
Here, we take a step towards resolving them by showing that the answer is indeed yes if we consider weak Muller conditions.

The games we consider here are Gale-Stewart games~\cite{GaleStewart53}, arguably the simplest form of two-player games.
In such a game, two players, called $I$ and $O$, alternatingly pick letters, thereby constructing an infinite word~$\alpha$.
The winner of such a play is determined by the winning condition of the game, a language of infinite words.
Player~$O$ wins if $\alpha$ is in the winning condition, otherwise Player~$I$ wins.

Delay games were introduced by Hosch and Landweber~\cite{DBLP:conf/icalp/HoschL72} only three years after the seminal Büchi-Landweber Theorem~\cite{BuechiLandweber69} showing how to determine the winner of Gale-Stewart games with $\omega$-regular winning conditions.
Hosch and Landweber generalized the setting of Gale-Stewart games by allowing Player~$O$ to delay her moves to obtain a lookahead on Player~$I$'s moves. 
This advantage allows her to win games she cannot win without lookahead.
These games capture the asynchronous interaction of two agents~\cite{DBLP:journals/acta/ChenFLMZ21}, the deterministic uniformization problem for relations over infinite words~\cite{carayol:hal-01806575,HKT12}, and are related to streaming transducers with delay~\cite{DBLP:conf/fsttcs/FiliotW21}.

In this work, we are interested in games with $\omega$-regular winning conditions.
These are typically represented by (deterministic or non-deterministic) $\omega$-automata with various kinds of acceptance conditions, e.g., safety, parity or Muller.
The choice between determinism and non-determinism and succinctness of the acceptance condition can have an influence on the complexity of solving the game.

Hosch and Landweber showed that it is decidable whether Player~$O$ wins a given delay game with bounded lookahead~\cite{DBLP:conf/icalp/HoschL72}. 
Forty years later, Holtmann, Kaiser, and Thomas~\cite{HKT12} revisited delay games and proved that if Player~$O$ wins a delay game then she wins it already with doubly-exponential lookahead (in the size of a given deterministic parity automaton recognizing the winning condition).
Thus, unbounded lookahead does not offer any advantage over doubly-exponential lookahead in games with $\omega$-regular winning conditions.
Furthermore, they showed that the winner of a delay game, again with its winning condition given by a deterministic parity automaton, can be determined in doubly-exponential time.

Both upper bounds were improved and matching lower bounds were proven by Klein and Zimmermann: Solving delay games is \exptime-complete and exponential lookahead is both necessary to win some games and sufficient to win all games that can be won~\cite{DBLP:journals/lmcs/KleinZ14}.
Both lower bounds already hold for winning conditions specified by deterministic safety automata while the upper bounds hold for deterministic parity automata.
They also considered reachability automata, for which there is no difference between the results for deterministic and non-deterministic automata. 

For non-deterministic parity automata, solving delay games is \mbox{\twoexp}-complete and doubly-ex\-po\-nen\-tial lookahead is both necessary and sufficient~\cite{DBLP:conf/fsttcs/KleinZ16}.
As before, both lower bounds already hold for non-deterministic safety automata while the upper bounds hold for non-deterministic parity automata.
See \cref{tab:complexity} and \cref{tab:lookahead} for an overview over the known results (in gray).

\setlength{\tabcolsep}{15pt}

\begin{table}[t]
\caption{The complexity of solving delay games. Results for reachability, safety, and parity (in gray) are from~\protect\cite{DBLP:journals/lmcs/KleinZ14,DBLP:conf/fsttcs/KleinZ16}; results for weak Muller are proven here.}
{
\begin{center}
\begin{tabular}{lll}
\toprule
& deterministic &  non-deterministic\\
 \midrule
\rowcolor{lightgray!40} Reachability &\pspace-complete &\pspace-complete \\
\rowcolor{lightgray!40} Safety &\exptime-complete &\twoexp-complete \\
\rowcolor{lightgray!40} Parity &\exptime-complete &\twoexp-complete \\
 weak Muller &\twoexp-complete &\threeexp-complete\\
\bottomrule
\end{tabular}
\end{center}
}
\label{tab:complexity}
\end{table}

\begin{table}[t]
\caption{The lookahead required to win delay games. All bounds are asymptotically tight. Results for reachability, safety, and parity (in gray) are from~\protect\cite{DBLP:journals/lmcs/KleinZ14,DBLP:conf/fsttcs/KleinZ16}; results for weak Muller are proven here.}
{
\begin{center}
\begin{tabular}{lll}
\toprule
& deterministic &  non-deterministic\\
 \midrule
\rowcolor{lightgray!40}  Reachability & exponential & doubly-exponential \\
\rowcolor{lightgray!40} Safety & exponential & doubly-exponential \\
\rowcolor{lightgray!40} Parity & exponential & doubly-exponential \\
weak Muller & doubly-exponential & triply-exponential \\
\bottomrule
\end{tabular}
\end{center}
}
\label{tab:lookahead}
\end{table}

However, the parity condition is not very succinct, e.g., deterministic Rabin, Streett, and Muller automata can all be exponentially more succinct than deterministic parity automata while they all recognize exactly the $\omega$-regular languages.

Here, we need to discuss briefly how to measure the size of an $\omega$-automaton, say with set~$Q$ of states.
For all acceptance conditions induced by a subset~$F\subseteq Q$ (e.g., reachability, safety, Büchi, and co-Büchi) or by a coloring~$\col \colon Q \rightarrow \nats$ (e.g., parity) the size of the automaton is captured by its number of states, as the acceptance condition can be encoded with a polynomial overhead.
However, for more succinct automata, the situation is different. 
Rabin and Streett conditions are given by a finite family~$(R_i,G_i)_{i\in\I}$ of pairs of subsets~$R_i, G_i \subseteq Q$.
Here, $\size{\I}$ is called the index of the automaton and has to be taken into account when measuring its size.

Muller acceptance on the other hand is based on a set~$\curlyF \subseteq 2^{Q}$, which can be encoded in various representations.
The simplest one is just listing all subsets in $\curlyF$, the so-called explicit representation whose size is $\size{\curlyF}$.
In this work we are mainly focused on representing~$\curlyF$ by a Boolean formula with variables in $Q$ so that the models of the formula are exactly the sets in $\curlyF$.
In this case, we measure the size of the acceptance condition by the size of the formula.
Other representations include the use of circuits, trees, DAGs, and colorings.
The relative succinctness of these representations and their influence on the complexity of solving arena-based (i.e., delay-free) games has been studied by Hunter and Dawar~\cite{DBLP:conf/mfcs/HunterD05} and Horn~\cite{DBLP:conf/fsttcs/Horn08}.

In the setting of arena-based games, parity games can be solved in quasi-polynomial time~\cite{DBLP:conf/stoc/CaludeJKL017} while solving Rabin games is \np-complete~\cite{DBLP:journals/siamcomp/EmersonJ99}, solving Streett games is \conp-complete (they are dual to Rabin games), and solving Muller games can be \ptime-complete, \conp-complete, or \pspace-complete, depending on the representation of $\curlyF$~\cite{DBLP:conf/mfcs/HunterD05,DBLP:conf/fsttcs/Horn08}.

Thus, it is natural to ask whether winning conditions given by more succinct automata also make delay games harder to solve and increase the bounds on the necessary lookahead for Player~$O$.
However, no such results are known. 
Note that it is straightforward to obtain doubly-exponential (triply-exponential) upper bounds on the complexity and on the lookahead, as every deterministic (non-deterministic) Rabin, Street, or Muller automaton can be turned into an equivalent exponentially larger deterministic (non-deterministic) parity automaton. 
This result is independent of the size and encoding of the acceptance condition of the automaton that is transformed.

Here, we are interested in Muller conditions, as they are the most succinct ones. 
As a first step towards answering our motivating questions, we show that for \emph{weak} Muller conditions, there is indeed an exponential increase, both in the solution complexity and in the necessary lookahead.
Thus, we provide matching lower bounds to the upper bounds obtained by transforming (weak) Muller automata into parity automata.

Recall that a Muller condition~$\curlyF \subseteq 2^Q$ represents all those runs~$\rho \in Q^\omega$ whose infinity set, the set of states visited infinitely often by $\rho$, is in $\curlyF$. 
In contrast, a weak Muller condition~$\curlyF \subseteq 2^Q$ represents all those runs~$\rho$ whose occurrence set, the set of states visited by $\rho$, is in $\curlyF$. 
Note that weak Muller automata are strictly less expressive than (standard) Muller automata.

In this paper we settle the complexity of solving delay games with deterministic and non-deterministic weak Muller winning conditions and provide tight bounds on the necessary lookahead to win such games.
More precisely, we show that solving delay games with deterministic weak Muller winning conditions is \twoexp-complete and solving delay games with non-deterministic weak Muller winning conditions is \threeexp-complete. 
Similarly, we show that doubly-exponential lookahead is necessary and sufficient to win delay games with deterministic weak Muller conditions and triply-exponential lookahead is necessary and sufficient to win delay games with non-deterministic weak Muller conditions.
All these results hold for $\curlyF$ being represented by a Boolean formula. 
The lower bounds rely on the fact that in a weak Muller automaton, every visit to a state is relevant for acceptance. 
We show that this property allows to construct a small deterministic (non-deterministic) automaton that requires the players in a game to implement a cyclic counter with exponentially (doubly-exponentially) many values. 
This automaton is then used to construct games in which Player~$O$ needs doubly-exponential (triply-exponential) lookahead and to construct games that simulate \aexpspace (\atwoexpspace) Turing machines.
Finally, we also give an exponential lower bound on the necessary lookahead in delay games with deterministic explicit weak Muller conditions. 

We conclude this paper by discussing the obstacles one faces when trying to take the next step, i.e., to raise the lower bounds for (standard) Muller conditions.
There, we also return to the discussion of the different representations of Muller conditions, i.e., we discuss whether our lower bounds, which hold for the representation by Boolean formulas, can be lifted to less succinct representations, thereby giving stronger results. 

\section{Preliminaries}
\label{sec:prelims}

We introduce our notation.
We denote the non-negative integers by~$\nats$, and the set $\set{0,1}$ by~$\bool$.
The size of a set $X$ is denoted by $\size{X}$.

\paragraph{Words, languages.}

An \emph{alphabet}~$\Sigma$ is a non-empty finite set of \emph{letters} or \emph{symbols}.
A \emph{word} over $\Sigma$ is a sequence of letters of $\Sigma$.
The set of finite resp.\ infinite words over $\Sigma$ is denoted by $\Sigma^*$ resp.\ $\Sigma^\omega$.
The set of non-empty finite words over $\Sigma$ is denoted by $\Sigma^+$.
The \emph{empty word} is denoted by $\varepsilon$.
Given a finite or infinite word $\alpha$, we denote by $\alpha(i)$ the $i$-th letter of~$\alpha$, starting with $0$, i.e., $\alpha = \alpha(0)\alpha(1)\alpha(2)\cdots$.
The length of $\alpha$ is denoted by~$|\alpha|$.
A subset $L \subseteq \Sigma^*$ resp.\ $L \subseteq \Sigma^\omega$ is a \emph{language} resp.\ an \emph{$\omega$-language} over $\Sigma$.
We drop the prefix $\omega$ when it is clear from the context.

Given two infinite words~$\alpha \in (\Sigma_0)^\omega$ and $\beta \in (\Sigma_1)^\omega$, we define $\binom{\alpha}{\beta} = \binom{\alpha(0)}{\beta(0)}\binom{\alpha(1)}{\beta(1)}\binom{\alpha(2)}{\beta(2)} \cdots \in (\Sigma_0 \times \Sigma_1)^\omega$. 

\paragraph{Automata.}

An \emph{$\omega$-automaton} is a tuple~$\aut = (Q, \Sigma, q_\initmark, \Delta, \acc)$ where $Q$ is a finite set of states, $\Sigma$ is an alphabet, $q_\initmark \in Q$ is an initial state, $\Delta \subseteq Q \times \Sigma \times Q$ is a transition relation, and $\acc \subseteq \Delta^\omega$ is a set of accepting runs.

An infinite \emph{run}~$\rho$ of $\aut$ is a sequence
\[
\rho = (q_0, a_0, q_1)(q_1, a_1, q_2)(q_2,a_2,q_3) \cdots \in \Delta^\omega.\]
As usual, we say that $\rho$ is initial if $q_0 = q_\initmark$ and we say that $\rho$ processes~$a_0a_1a_2\cdots \in \Sigma^\omega$. 
Given a run $\rho$, let $\infi(\rho)$ be the set of states visited infinitely often by $\rho$, and let $\occ(\rho)$ be the set of states visited by~$\rho$.

The language \emph{recognized} by $\aut$ is the set $L(\aut) \subseteq \Sigma^\omega$ that contains all $\omega$-words that have an initial run in $\acc$ processing it.

An \emph{$\omega$-automaton} $\aut = (Q, \Sigma, q_\initmark, \Delta, \acc)$ is \emph{deterministic} (and complete) if $\Delta$ is given as a total function $\delta\colon Q \times \Sigma \to Q$.
If we speak of \emph{the} run of $\aut$ on $\alpha \in \Sigma^\omega$, then we mean the unique initial run of $\aut$ processing~$\alpha$. 

\paragraph{Acceptance conditions.}

We recall safety, reachability, parity, and (weak) Muller acceptance conditions.
%

An $\omega$-automaton~$\aut = ( Q, \Sigma, q_\initmark, \Delta, \acc )$ is a \emph{safety automaton}, if there is a set~$F  \subseteq Q$ of accepting states such that 
\[\acc = \set{(q_0, a_0, q_1) (q_1, a_1, q_2) (q_2, a_2, q_3) \cdots \in \Delta^\omega \mid q_i \in F \text{ for every } i}.\]

Moreover, an $\omega$-automaton~$\aut = ( Q, \Sigma, q_\initmark, \Delta, \acc )$ is a \emph{reachability automaton}, if there is a set~$F  \subseteq Q$ of accepting states such that 
\[\acc = \set{(q_0, a_0, q_1) (q_1, a_1, q_2) (q_2, a_2, q_3) \cdots \in \Delta^\omega \mid q_i \in F \text{ for some } i}.\]

Furthermore, an $\omega$-automaton~$\aut = ( Q, \Sigma, q_\initmark, \Delta, \acc )$ is a \emph{parity automaton}, if there is some coloring~$\col \colon Q \rightarrow \nats$ such that
\[\acc = \set{\rho \in \Delta^\omega \mid \text{$\max \set{ \col(q) \mid q \in \infi(\rho) }$ is even}}.\]
In the remainder, we use the acronyms~\DPA resp.\ \NPA to refer to deterministic parity automata resp.\ non-deterministic parity automata.

Finally, an $\omega$-automaton~$\aut = ( Q, \Sigma, q_\initmark, \Delta, \acc )$ is a \emph{Muller automaton} resp.\ a \emph{weak Muller automaton}, if there is a family~$\curlyF \subseteq \pow{Q}$ of sets of states such that 
\[\acc = \set{ \rho \in \Delta^\omega \mid \infi(\rho) \in \curlyF}\] 
resp.\ 
\[\acc = \set{ \rho \in \Delta^\omega \mid \occ(\rho)~\in~\curlyF}.\]
In the remainder, we use the acronyms~\wDMA resp.\ \wNMA to refer to deterministic weak Muller automata resp.\ non-deterministic weak Muller automata.

The sets $\acc$ are infinite objects, but they have finite representations.
E.g., for safety conditions $\acc$ is fully specified by a subset~$F \subseteq Q$ of safe states; for (weak) Muller conditions $\acc$ is fully specified by a set~$\curlyF \subseteq \pow{Q}$ of sets of states.
While these representations are usually small in the size of the automaton, this is not guaranteed for Muller acceptance conditions.
Working with Muller automata, the representations of the acceptance condition influence the complexity of  decision problems~\cite{DBLP:conf/mfcs/HunterD05,HunterPhD}.
Hence, we recall different ways of representing Muller conditions which are more succinct than simply listing the elements of $\curlyF$.

\paragraph{Representations of acceptance conditions.}

Hunter and Dawar studied the impact of the representation of Muller conditions on the complexity of solving delay-free Muller games~\cite{DBLP:conf/mfcs/HunterD05,HunterPhD}. 
In this paper, we follow their terminology and definitions for these representations.

The most straightforward representation of an acceptance condition of a Muller automaton is an \emph{explicit condition}, i.e., an explicit list of the elements of $\curlyF$.
For short, we refer to a \wDMA resp.\ \wNMA where the acceptance condition is represented as an explicit condition as \explicit-\wDMA resp.\ \explicit-\wNMA.

But there are more succinct representations of Muller automata acceptance conditions, such as \emph{Muller conditions over colors}, \emph{win-set conditions}, \emph{Emerson-Lei conditions}, and \emph{circuit conditions} (see \cite{DBLP:conf/mfcs/HunterD05,HunterPhD} for definitions and for a comparison in the setting of delay-free games).

In this work, we consider Emerson-Lei conditions which are Boolean formulas $\varphi$ with variables from $Q$.
The set $\curlyF_{\varphi}$ specified by $\varphi$ is the set of sets~$F \subseteq Q$ such that the truth assignment that maps each element of $F$ to true and each element of $Q \setminus F$ to false satisfies $\varphi$.
The size of a formula~$\varphi$ is denoted $\size{\varphi}$ and is defined as usual as its length.
For short, we refer to a \wDMA resp.\ \wNMA where the acceptance condition is represented as an Emerson-Lei condition as an \formula-\wDMA resp.\ an \formula-\wNMA.

Let us conclude by mentioning that Emerson-Lei conditions are more succinct than Muller conditions over colors, which are more succinct than win-set conditions, which are more succinct than explicit conditions.
Circuit conditions are at least as succinct as Emerson-Lei conditions, but it is open if circuit conditions are more succinct than Emerson-Lei conditions (meaning it is open whether circuits can always be translated into small formulas).
These (and further) results can be found in \cite{DBLP:conf/mfcs/HunterD05,HunterPhD}.

\paragraph{Size of automata.}

As discussed in the previous paragraph, the representation of the acceptance condition of a (weak) Muller automaton can be more or less succinct.
This should be reflected when defining the size of a (weak) Muller automaton.
Hence, we define the size of a (weak) Muller automaton with Emerson-Lei condition~$\varphi$ as $|Q| + |\varphi|$; the size of a (weak) Muller automaton with explicit condition~$\curlyF$ as $|Q| + |\curlyF|$. 

For other $\omega$-automata acceptance conditions such as, e.g., reachability, safety or parity, we define the size of the automaton as its number of states, as the acceptance condition can be encoded with a polynomial overhead.

\paragraph{Delay games.}

A \emph{delay function} is a mapping~$f \colon \nats \rightarrow \nats_{\geq 1}$, which is said to be constant if $f(i) =1$ for all $i>0$. A \emph{delay game}~$\delaygame{L}$ consists of a delay function~$f$ and a \emph{winning condition}~$L \subseteq (\SigmaI \times \SigmaO)^\omega$ for some alphabets~$\SigmaI$ and $\SigmaO$. Such a game is played in rounds~$i = 0,1,2, \ldots$ as follows: in round~$i$, first Player~$I$ picks a word~$u_i \in \SigmaI^{f(i)}$, then Player~$O$ picks a letter~$v_i \in \SigmaO$. Player~$O$ \emph{wins} a play~$(u_0, v_0)(u_1, v_1)(u_2, v_2) \cdots $ if the outcome~$\binom{ u_0 u_1 u_2 \cdots }{ v_0 v_1 v_2 \cdots }$ is in $L$; otherwise, Player~$I$ wins.
Note that if $f$ is a constant delay function, then Player~$O$ has a constant lookahead of $f(0)$ letters on her opponents moves.
Hence, the moniker constant refers to the lookahead induced by $f$, not to $f$ itself.

A \emph{strategy} for Player~$I$ in $\delaygame{L}$ is a mapping $\stratI \colon \SigmaO^* \rightarrow \SigmaI^*$ satisfying $\size{\stratI(w)} = f(\size{w})$ while a strategy for Player~$O$ is a mapping~$\stratO \colon \SigmaI^+ \rightarrow \SigmaO$. A play~$(u_0, v_0)(u_1, v_1)(u_2, v_2) \cdots $ is \emph{consistent} with $\stratI$ if $u_i = \stratI(v_0 \cdots v_{i-1})$ for all $i$, and it is consistent with $\stratO$ if $v_i = \stratO(u_0 \cdots u_i)$ for all $i$. A strategy for Player~$\p \in \set{I,O}$ is \emph{winning}, if every play that is consistent with the strategy is won by Player~$\p$.
We say that Player~$\p\in \set{I,O}$ \emph{wins} a game  $\delaygame{L}$ if Player~$\p$ has a winning strategy in $\delaygame{L}$.

\paragraph{Problem statement.}

In this work, we investigate delay games with weak Muller automata winning conditions.
We are interested in two aspects.

Firstly, what is the computational complexity of deciding whether Player~$O$ wins such a game?
Secondly, how much lookahead does Player~$O$ need to win such a game if possible?

Note that delay games with \wDMA resp.\ \wNMA winning conditions are determined, because delay games with Borel winning conditions are determined \cite{DBLP:journals/csl/Klein015}.
Hence, we also speak of deciding whether Player~$O$ wins such a game as solving such a game.

We end this section with some introductory examples.

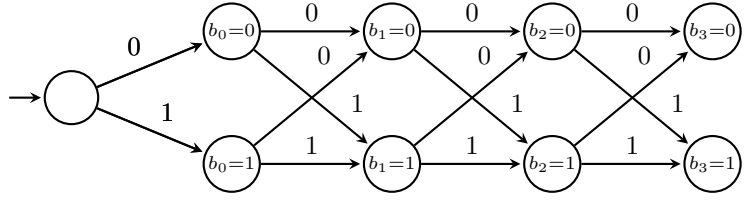
\begin{figure}[t]
  \hspace*{\fill}
  \begin{tikzpicture}[thick]
    \tikzstyle{every state}+=[inner sep=1pt, minimum size=2em];

    \node[state, initial] (s) {};

    \node[state, right of = s, yshift = 2.5em] (00) {\nodelbl{0}{0}};
    \node[state, right of = 00] (10) {\nodelbl{1}{0}};
    \node[state, right of = 10] (20) {\nodelbl{2}{0}};
    \node[state, right of = 20] (30) {\nodelbl{3}{0}};

    \node[state, right of = s, yshift = -2.5em] (01) {\nodelbl{0}{1}};
    \node[state, right of = 01] (11) {\nodelbl{1}{1}};
    \node[state, right of = 11] (21) {\nodelbl{2}{1}};
    \node[state, right of = 21] (31) {\nodelbl{3}{1}};

    \draw[->] (s) edge[auto] node {0} (00);
    \draw[->] (00) edge[auto] node {0} (10);
    \draw[->] (10) edge[auto] node {0} (20);
    \draw[->] (20) edge[auto] node {0} (30);

    \draw[->] (s) edge[auto] node {1} (01);
    \draw[->] (00) edge[auto,near end] node {1} (11);
    \draw[->] (10) edge[auto,near end] node {1} (21);
    \draw[->] (20) edge[auto,near end] node {1} (31);

    \draw[->] (s) edge[auto] node {0} (00);
    \draw[->] (01) edge[auto,near end] node {0} (10);
    \draw[->] (11) edge[auto,near end] node {0} (20);
    \draw[->] (21) edge[auto,near end] node {0} (30);

    \draw[->] (s) edge[auto] node {1} (01);
    \draw[->] (01) edge[auto] node {1} (11);
    \draw[->] (11) edge[auto] node {1} (21);
    \draw[->] (21) edge[auto] node {1} (31);
  \end{tikzpicture}
  \hspace*{\fill}
  \caption{A gadget to store a $4$-bit sequence.}
  \label{fig:gadget-bit-store}
\end{figure}


\begin{example}\label[ex]{ex:gadget}
  In \cref{fig:gadget-bit-store}, we present a deterministic automaton gadget whose structure is suitable to store a 4-bit sequence.
  What we mean by that is, that for a word starting with $a_0a_1a_2a_3 \in \bool^4$ the corresponding run $(p_0, a_0, p_1)(p_1, a_1, p_2)(p_2, a_2, p_3)(p_3, a_3, p_4) \cdots$ of the automaton yields that the state $p_{i+1}$ is equal to $\nodelbl{i}{a_i}$ for $0 \leq i \leq 3$.
  This idea is easily scaled to an $n$-bit sequence for arbitrary $n \in \nats_{\geq 1}$.
  Note that the gadget is of linear size in~$n$.
\end{example}

Now, we present a delay game which makes use of such a gadget.

\begin{example}\label[ex]{ex:gadgetdelay}
  Let $L \subseteq (\bool^2)^\omega$ be the $\omega$-language that contains all words of the form $(\alpha,\beta)$ such that $\beta(i) = \alpha(4+i)$ for $0 \leq i \leq 3$.
  In words, the second four-bit segment of the input component is equal to the first four-bit segment of the output component.

  \newcommand{\first}[2]{\mathit{G}_1(\nodelbl{#1}{#2})}
  \newcommand{\second}[2]{\mathit{G}_2(\nodelbl{#1}{#2})}

  We describe a \wDMA $\aut$ that recognizes $L$.
  Basically, such an automaton is composed of two successive gadgets similar as in \cref{fig:gadget-bit-store} and a sink state reached after the second gadget.
  In the first gadget, the output component of a letter determines the target state, and in the second gadget, the input component of a letter determines the target state.
  We assume that the states of the first gadget are called $\first{0}{0}$, $\first{0}{1}$, etc., and in the second gadget $\second{0}{0}$, $\second{0}{1}$, etc.

  We give the acceptance condition of $\aut$ as a formula $\varphi$:
  \[
    \begin{array}{ll}
       & (\first{0}{0} \leftrightarrow \second{0}{0}) \\
      \wedge & (\first{1}{0} \leftrightarrow \second{1}{0}) \\
      \wedge & (\first{2}{0} \leftrightarrow \second{2}{0}) \\
      \wedge & (\first{3}{0} \leftrightarrow \second{3}{0})
    \end{array}
  \]
  Clearly, $L(\aut) = L$.
  
  Now, let us consider a delay game with $L$ as winning condition.
  Player~$O$ can win such a game if she is aware of the $(4+k)$-th letter that Player~$I$ plays before she has to give her $k$-th letter (for $1 \leq k \leq 4$).
  The constant delay function $f$ with $f(0) = 5$ gives enough lookahead to ensure this.

We show that Player~$O$ wins $\delaygame{L}$.
Due to the choice of $f(0)$, in each round~$i$, Player~$I$ has produced a prefix~$\alpha(0) \cdots \alpha(i+4)$ that Player~$O$ can base her move in that round on. 
Hence, she picks $\beta(i) = \alpha(i+4)$ in round~$i$.
This strategy is winning for her, as every consistent outcome is in $L$.

  As for \cref{ex:gadget}, this example is easily scaled to an $n$-bit sequence for arbitrary $n \in \nats_{\geq 1}$.
  Note that a formula $\varphi_n$ specifying the acceptance condition is of size $\bigo(n)$ while an explicit representation of $\curlyF_{\varphi_n}$ is of size~$\bigo(2^n)$ as all possible $n$-bit sequences must be explicitly represented.
\end{example}

In the above example, we presented a delay game where a small lookahead is sufficient for Player~$O$ to win.
We now give an example of a family of $\omega$-automata where, when used as a winning condition for a delay game, Player~$O$ needs exponential lookahead (in the size of the automaton) to win.

\begin{figure}[t]
  \hspace*{\fill}
  \begin{tikzpicture}[thick]
    \tikzstyle{every state}+=[inner sep=3pt, minimum size=1em];
 
    \tikzstyle{small}=[scale=1,draw,gray];
    \tikzstyle{scaled}=[scale=1];
    
    \node[state, small] (c) {};
    \node[state, small, right of = c]         (d) {};
    \node[state, small, below of = c] 				(e) {};
    \node[state, small, right of = e]         (f) {};
    
    \node[state, initial, left of = c, yshift=-3em,xshift=-2em, initial text = $\autp_n$]  (a) {$q_\initmark$};
    \node[state, right of = d, yshift=-3em,xshift=2em] (b) {$q_f$};
    
    \node[scale=0.8,gray] at ($(c)+(-0.5,0.55)$) {$\mathfrak G_0$};
    \node[scale=0.8,gray] at ($(e)+(-0.5,0.55)$) {$\mathfrak G_{n-1}$};
    
    \node[gray,yshift=0.75em] at ($(a)!0.5!(b)$) {$\vdots$};
    
    \draw[->,gray]
      (c) edge[bend left=15] node {} (d)
      (c) edge[loop above]   node {} ()
      (d) edge[loop above]   node {} ()
      (d) edge[bend left=15] node {} (c);
    
    \draw[->,gray]
      (e) edge[bend left=15] node {} (f)
      (e) edge[loop above]   node {} ()
      (f) edge[loop above]   node {} ()
      (f) edge[bend left=15] node {} (e);

    \draw[->]
      (a) edge node[scaled] 
        {$\binom{\ast}{0}$} (c)
      (a) edge node[scaled,swap] 
        {$\binom{\ast}{n-1}$} (e)
      (d) edge node[scaled] 
        {$\binom{0}{\ast}$} (b)
      (f) edge node[scaled,swap] 
        {$\binom{n-1}{\ast}$} (b)
      (b) edge[loop right] node[scaled]
        {$\binom{\ast}{\ast}$} ();
    
    \draw[rounded corners,gray,dashed] ($(c)-(0.9,0.5)$) rectangle ($(d)+(0.9,0.8)$) {};
    \draw[rounded corners,gray,dashed] ($(e)-(0.9,0.5)$) rectangle ($(f)+(0.9,0.8)$) {};
    \end{tikzpicture}
    \hspace*{\fill}

\phantom{bla}

    \hspace*{\fill}
    \begin{tikzpicture}[thick]
      \tikzstyle{every state}+=[inner sep=1pt, minimum size=1em];
      \tikzstyle{scaled}=[scale=1];
      
      \node[state]  at (0,0) (0) {};
      \node[state]  at (2,0) (1) {};
      \node at ($(0)+(-0.75,0.75)$) {$\mathfrak G_j$};
      
      \draw[->]
        (0) edge[bend left=15] node[scaled] 
        {$\binom{j}{\ast}$} (1)
        (0) edge[loop above]   node[scaled] 
        {$\binom{\neq j}{\ast}$} ()
        (1) edge[loop above]   node[scaled]  
        {$\binom{< j}{\ast}$}()
        (1) edge[bend left=15] node[scaled] 
        {$\binom{>j}{\ast}$}(0);

      \draw[rounded corners,dashed] ($(0)-(1.2,-1.25)$) rectangle ($(1)+(1.2,-1)$) {};	  
    \end{tikzpicture}
    \hspace*{\fill}
    \caption{
     The automaton $\autp_n$ (top) contains gadgets $\mathfrak G_0,\dots,\mathfrak G_{n-1}$ (bottom).
     Transitions not depicted lead to a sink state, which is not drawn.
     Here, $\ast$ denotes an arbitrary letter from the respective alphabet.
    }
  \label{fig:jpair}
\end{figure}
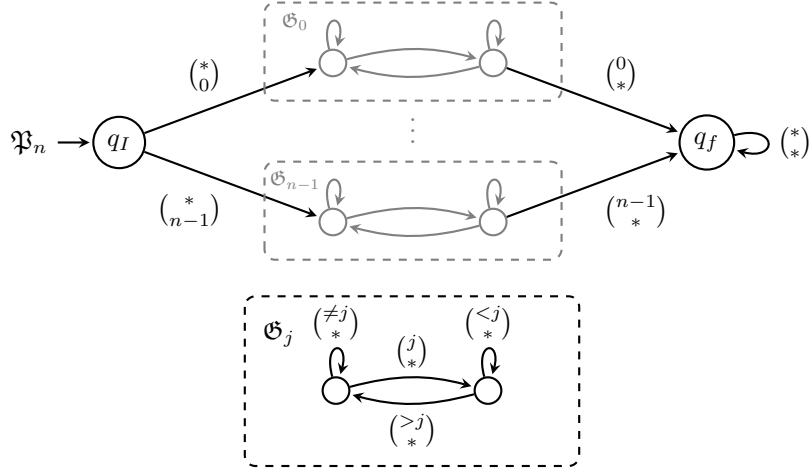

\begin{example}\label[ex]{ex:jpair}
Pick some $n \in \nats_{\geq 1}$, and let $\Sigma_n = \set{0,\ldots,n-1}$.
A sequence $\alpha \in (\Sigma_n)^* \cup (\Sigma_n)^\omega$ is said to contain a so-called \emph{bad $j$-pair} if there are two positions $p<p'$ such that $\alpha(p) = \alpha(p') = j$ and $\alpha(q) < j$ for all $p < q < p'$.\footnote{\label{footnote:jpairs}Note that we could equivalently require $\alpha(q) \le j$ for all $p < q < p'$. The latter allows occurrences of $j$ between positions $p$ and $p'$. However, one can then pick two such occurrences without any $j$'s in between, which satisfy the stricter definition with $<$.}

Now consider the language $P_n \subseteq (\Sigma_n \times \Sigma_n)^\omega$ that contains a word~$\binom{\alpha}{\beta}$ if, and only if, $\alpha(1)\alpha(2)\cdots$ contains a bad $j$-pair where $j = \beta(0)$.
In words, if the first output letter is $j$, then there is a bad $j$-pair in the input letter stream starting from the second letter. 

We have that $P_n = L(\autp_n)$ where $\autp_n$ is the automaton specified in \cref{fig:jpair} when it is, for example, interpreted as a reachability automaton where $q_f$ has to be reached.
Note that $\autp_n$ is deterministic and of size~$\bigo(n)$.
\end{example}

Klein and Zimmermann~\cite{DBLP:journals/lmcs/KleinZ14} have proven that in order to be guaranteed the existence of a bad $j$-pair for $0 \leq j \leq n-1$ a sequence over $\Sigma_n$ of length at least $2^n$ is necessary.

\begin{proposition}\label[prop]{prop:jpair}
  Let $n \in \nats_{\geq 1}$, and let $\Sigma_n = \set{0,\ldots,n-1}$.
  \begin{itemize}
    \item Every word $w \in (\Sigma_n)^*$ with $\size{w} \geq 2^n$ contains a bad $j$-pair for some $0 \leq j \leq n-1$.
    \item There is a word $w \in (\Sigma_n)^*$ with $\size{w} = 2^n - 1$ that does not contain a bad $j$-pair for every $0 \leq j \leq n-1$.
  \end{itemize}
\end{proposition}

With this in mind, it is intuitively true that in a delay game with winning condition~$P_n$ as given in \cref{ex:jpair}, Player~$O$ can always win if she has at least a lookahead of $2^n$ on Player~$I$'s moves, but not with less lookahead.
This was formally proven in~\cite{DBLP:journals/lmcs/KleinZ14}.
In our upcoming proofs, we will incorporate variations of \cref{ex:jpair}.

To conclude, we show that the game presented in \cref{ex:jpair} implies an exponential lower bound result on the necessary lookahead for Player~$O$ to win delay games with deterministic weak Muller automata winning conditions which are stated explicitly.

The following results is obtained by picking $L_n = P_n$ and showing that $\autp_n$ is an $\explicit$-$\wDMA$ with set $\curlyF = \set{F_0,\ldots,F_{n-1}}$, where the set~$F_i$ contains $q_\initmark$ and $q_f$ as well as all (that is, the two states) from gadget $\mathfrak{G}_i$.
  To see this, consider a run that starts in $q_\initmark$ and reaches the sink~$q_f$: only one of the gadgets $\mathfrak{G_i}$ is visited.
  Also, both states of $\mathfrak{G_i}$ must be visited.

\begin{corollary}
  \label{cor:detmullerLowerboundLookaheadExplicit}
  For every $n \in \nats_{\geq 1}$, there exists a language $L_n$ recognized by an \explicit-\wDMA~$\aut_n$ of size~$\bigo(n)$ such that
  \begin{itemize}
    \item Player~$O$ wins $\delaygame{L_n}$ for some constant delay function $f$, but
    \item Player~$I$ wins $\delaygamep{L_n}$ for every delay function $g$ with $g(0) \leq 2^{n}$.
  \end{itemize}
\end{corollary}


\section{Upper Bounds for Weak Muller Delay Games}
\label{sec:upperbounds}

We begin this section with a recap of known results used to prove our upper bounds on complexity and necessary lookahead.
Klein and Zimmermann~\cite{DBLP:journals/lmcs/KleinZ14,DBLP:conf/fsttcs/KleinZ16} have shown the following complexity results about parity delay games.

\begin{proposition}\label[prop]{prop:parity}
  \phantom{abc}
  \begin{enumerate}
    \item Solving delay games with \DPA winning conditions is \exptime-complete.
    
    \item Solving delay games with \NPA winning conditions is \twoexp-complete.
  \end{enumerate}
\end{proposition}

Furthermore, Klein and Zimmermann~\cite{DBLP:journals/lmcs/KleinZ14,DBLP:conf/fsttcs/KleinZ16} have shown tight bounds on the necessary lookahead needed to win parity delay games.
Here, we are interested in the upper bound results.

\begin{proposition}\label[prop]{prop:paritylookahead}
  \phantom{abc}
  \begin{enumerate}
    \item For every \DPA $\aut$, the following are equivalent:
    \begin{itemize}
      \item Player~$O$ wins $\delaygame{L(\aut)}$ for some constant delay function $f$ where $f(0)$ is exponential in the size of $\aut$. 
      \item Player~$O$ wins $\delaygamep{L(\aut)}$ for some delay function $g$.
    \end{itemize}
    \item For every \NPA $\aut$, the following are equivalent:
    \begin{itemize}
      \item Player~$O$ wins $\delaygame{L(\aut)}$ for some constant delay function $f$ where $f(0)$ is doubly-exponential in the size of $\aut$. 
      \item Player~$O$ wins $\delaygamep{L(\aut)}$ for some delay function $g$.
    \end{itemize}
  \end{enumerate} 
\end{proposition}

To obtain our upper bounds on weak Muller delay games stated below (see \cref{thm:upperbound,thm:upperboundLookahead}), we explicitly state the following easy result on the translation from weak Muller into parity automata.

\begin{lemma}\label[lem]{lemma:main}
  \phantom{abc}
  \begin{enumerate}
    \item For every \wDMA (in any representation) there exists an equivalent exponentially-sized \DPA.
    
    \item For every \wNMA (in any representation) there exists an equivalent exponentially-sized \NPA.
\end{enumerate}
  \end{lemma}

  \begin{proof}
    Given a weak Muller automaton~$\aut$, it suffices to construct a parity automaton~$\aut'$ that additionally tracks the occurrence set of a run of~$\aut$.
    To this end, the parity automaton uses the state set~$Q \times 2^Q$ where $Q$ is the set of states of $\aut$. 
    The first component simulates a run of $\aut$ while the second one accumulates the occurrence set of the run prefix simulated thus far.
    
    States of the form $(q, O)$ with $O \in \curlyF$ are assigned color~$2$ and those with $O \notin \curlyF$ are assigned color~$1$. 
    The resulting parity automaton is equivalent to the original weak Muller automaton, as the occurrence set eventually stabilizes.
  \end{proof}

We are ready to prove our upper bounds (regarding complexity and necessary lookahead) for delay games with weak Muller automata winning conditions.

  The combination of \cref{prop:parity,lemma:main} immediately yields our first theorem which states upper bound complexity results.
  
  \begin{theorem}\label{thm:upperbound}
    \phantom{abc}
  \begin{enumerate}
      \item Solving delay games with \wDMA winning conditions (in any representation) is in \twoexp.
      
      \item Solving delay games with \wNMA winning conditions (in any representation) is in \mbox{\threeexp}.
  \end{enumerate}
  \end{theorem}

  The combination of \cref{prop:paritylookahead,lemma:main} immediately yields our next theorem which gives upper bounds on the necessary lookahead for Player~$O$ to win delay games.
  
\begin{theorem}\label{thm:upperboundLookahead}
  \phantom{abc}
  \begin{enumerate}
    \item For every \wDMA $\aut$ (in any representation), the following are equivalent:
    \begin{itemize}
      \item Player~$O$ wins $\delaygame{L(\aut)}$ for some constant delay function $f$ where $f(0)$ is doubly-exponential in the number of states of $\aut$.    
      \item Player~$O$ wins $\delaygamep{L(\aut)}$ for some delay function $g$.
    \end{itemize}
    \item For every \wNMA $\aut$ (in any representation), the following are equivalent:
    \begin{itemize}
      \item Player~$O$ wins $\delaygame{L(\aut)}$ for some constant delay function $f$ where $f(0)$ is triply-exponential in the number of states of $\aut$.    
      \item Player~$O$ wins $\delaygamep{L(\aut)}$ for some delay function $g$.
    \end{itemize}
  \end{enumerate}
 \end{theorem}

  In \cref{sec:detmuller,sec:nondetmuller} we show matching lower bounds for \cref{thm:upperbound,thm:upperboundLookahead}.

\section{Lower Bounds for Deterministic Weak Muller Delay Games}
\label{sec:detmuller}

This section is devoted to showing lower bounds for delay games with deterministic weak Muller automata winning conditions.

\paragraph{Lookahead.}

We begin this section by showing a doubly-exponential lower bound on the necessary lookahead for Player~$O$ to win, which yields a tight bound in combination with \cref{thm:upperboundLookahead}.

Recall \cref{ex:jpair}, where the concept of bad $j$-pairs was introduced, and an automaton of size $\bigo(n)$ was given where the first letter, say $i_0 \in [0,n-1]$ of the output component indicates the existence of a bad $i_0$-pair in the input component (starting from the second letter) which is a word over $[0,n-1]$.
  Recall that according to \cref{prop:jpair}, one needs exponential lookahead in $n$ to correctly identify a bad $j$-pair in a word over $[0,n-1]$.

  We are going to design a variant where a doubly-exponential lookahead is necessary.
To this end, we encode the numbers in the range~$[0,2^n-1]$ in binary and show how to construct a small automaton that checks whether the input component has a bad $i_0$-pair, where $i_0$ is the first number in the output component. 
This results in a game that Player~$O$ can only win with doubly-exponential lookahead.

In \cref{ex:jpair}, the automaton stores $i_0$ in its state space, as there are only $n$ possibilities.
However, this is no longer feasible with the range~$[0,2^n-1]$ and a polynomially-sized automaton, as there are $2^n$ possible values for $i_0$.
Instead, Player~$O$ has to mark two numbers she claims to form a bad $i_0$-pair and the automaton then checks whether the two marked numbers are equal to $i_0$ and whether all numbers in between are strictly smaller than $i_0$.
Using the succinctness of the Emerson-Lei condition, this can be achieved with a linearly-sized automaton.

\begin{theorem}\label{thm:detmullerLowerboundLookahead}
  For every $n \in \nats_{\geq 1}$, there exists a language $L_n$ recognized by an \formula-\wDMA $\aut_n$ of size $\bigo(n)$ such that
  \begin{itemize}
    \item Player~$O$ wins $\delaygame{L_n}$ for some constant delay function $f$, but
    \item Player~$I$ wins $\delaygamep{L_n}$ for every delay function $g$ with $\sum_{i = 0}^{n-1} g(i) \leq 2^{2^n}$.
  \end{itemize}
\end{theorem}

\begin{proof}
  Pick some $n \in \nats_{\geq 1}$.
  Our goal is to give a language~$L_n$ recognized by an \formula-\wDMA~$\aut_n$ of size~$\bigo(n)$ such that the input component encodes a sequence~\mbox{$x_0, x_1, x_2,\ldots$} of numbers~$x_i \in [0,2^n-1]$, and the output component encodes a sequence~\mbox{$y_0,y_1,y_2,\ldots$} of numbers~$y_i \in [0,2^n-1]$. Each $x_i$ and $y_i$ is encoded as an $n$-bit sequence. For acceptance, we require that $x_1x_2\cdots$ contains a bad $y_0$-pair.
  According to \cref{prop:jpair}, finding a bad $j$-pair in a word over $[0,2^n-1]$ requires in the worst case a word of length~$2^{2^n}$.
  Consequently, using $L_n$ as the winning condition for a delay game will ensure that a doubly-exponential lookahead (in $n$) is necessary.
  
  We go into details.
  The automaton we construct uses the product alphabet~$\set{0,1,\#} \times \set{0,1,\yes,\#}$.
  We first describe what we call valid input resp.\ output encodings:
  \begin{itemize}
    \item An input sequence $\alpha$ is a valid encoding if it is of the form 
    \[
        ((0+1)^n\#)^\omega.
    \]
    The $i$-th $n$-bit sequence encodes the number $x_i \in [0,2^n-~1]$, the least significant bit is the leftmost one.
    \item An output sequence $\beta$ is a valid encoding if it is of the form 
    \[
        ((0+1)^n(\yes+\#))^\omega
    \]
    such that $\yes$ occurs exactly twice.
    The $i$-th $n$-bit sequence  encodes the number $y_i \in [0,2^n-~1]$, the least significant bit is the leftmost one.
  \end{itemize}

The behavior of the automaton is described below, where we distinguish three cases.
The first two cases simply state that (viewed as a game) the player who first violates the correct encoding format loses.
\begin{enumerate}
  \item\label[item]{item:i1} If there exists some $i$ such that $\alpha(0)\alpha(1)\cdots\alpha(i)$ cannot be completed to a valid input encoding, but $\beta(0)\beta(1)\cdots\beta(i-1)$ can be completed to a valid output encoding, the automaton accepts.
  \item\label[item]{item:i2} If there exists some $i$ such that $\beta(0)\beta(1)\cdots\beta(i)$ cannot be completed to a valid output encoding, but $\alpha(0)\alpha(1)\cdots\alpha(i)$ can be completed to a valid input encoding, the automaton rejects.
  \item\label[item]{item:i3} If $\alpha$ and $\beta$ are valid encodings, the automaton accepts if a bad $y_0$-pair is correctly marked.
\end{enumerate}

\begin{figure}[t]
  \hspace*{\fill}
  \begin{tikzpicture}[thick,scale=0.8]
    \tikzstyle{every state}+=[inner sep=1pt, minimum size=2em];

    \node[state, initial] (s) {};

    \node[state, rectangle, rounded corners, right of = s, yshift = 6em] (0a) {\res{0}{0}{0}};
    \node[state, rectangle, rounded corners, right of = s, yshift = 3em] (0b) {\res{0}{0}{1}};
    \node[state, rectangle, rounded corners, right of = s, yshift = -3em] (0c) {\res{0}{1}{0}};
    \node[draw=gray,state, rectangle, rounded corners, right of = s, yshift = -6em] (0d) {\textcolor{gray}{\res{0}{1}{1}}};

    \node[state, rectangle, rounded corners, right of = 0a] (1a) {\res{1}{0}{0}};
    \node[state, rectangle, rounded corners, right of = 0b] (1b) {\res{1}{0}{1}};
    \node[state, rectangle, rounded corners, right of = 0c] (1c) {\res{1}{1}{0}};
    \node[state, rectangle, rounded corners, right of = 0d] (1d) {\res{1}{1}{1}};

    \node[draw = none, right of = s, xshift = 10em] (m) {$\cdots$};

    \node[state, rectangle, rounded corners, right of = 0a, xshift = 10em] (na) {\res{n-1}{0}{0}};
    \node[state, rectangle, rounded corners, right of = 0b, xshift = 10em] (nb) {\res{n-1}{0}{1}};
    \node[state, rectangle, rounded corners, right of = 0c, xshift = 10em] (nc) {\res{n-1}{1}{0}};
    \node[state, rectangle, rounded corners, right of = 0d, xshift = 10em] (nd) {\res{n-1}{1}{1}};

    \draw[->] (s) edge[auto] node {$\scriptstyle\binom{0}{0}$} (0a);
    \draw[->] (s) edge[auto, near end, swap] node {$\scriptstyle\binom{0}{1},\binom{1}{0}$} (0c);
    \draw[->] (s) edge[auto, near end, swap] node {$\scriptstyle\binom{1}{1}$} (0b);

    \draw[->] (0b) edge[auto] node {$\scriptstyle\binom{0}{0}$} (1c);
    \draw[->] (0b) edge[auto, swap] node[yshift=1.5em] {$\scriptstyle\binom{0}{1},\binom{1}{0},\binom{1}{1}$} (1d);

    \draw[->,gray,dashed] (0a) edge[auto] node {} ($(0a)+(4em,0em)$);
    \draw[->,gray,dashed] (0c) edge[auto] node {} ($(0c)+(4em,0em)$);
    \draw[->,gray,dashed] (0d) edge[auto] node {} ($(0d)+(4em,0em)$);

    \draw[->,gray,dashed] (1a) edge[auto] node {} ($(1a)+(4em,0em)$);
    \draw[->,gray,dashed] (1b) edge[auto] node {} ($(1b)+(4em,0em)$);
    \draw[->,gray,dashed] (1c) edge[auto] node {} ($(1c)+(4em,0em)$);
    \draw[->,gray,dashed] (1d) edge[auto] node {} ($(1d)+(4em,0em)$);

  \end{tikzpicture}
  \hspace*{\fill}
  \caption{A gadget to add two $n$-bit sequences in little endian notation.
  Missing transitions are implied by gray dashed arrows. States for positions 2 to $n-2$ are not drawn.
  The state label $\res{k}{b}{b'}$ indicates that adding the bits at position $k$ yields the result $b$ and $b'$ is the carry bit relevant for the addition at position $k+1$.
  The state $\res{0}{1}{1}$ is not reachable and can be omitted.}
  \label{fig:gadget-add}
\end{figure}
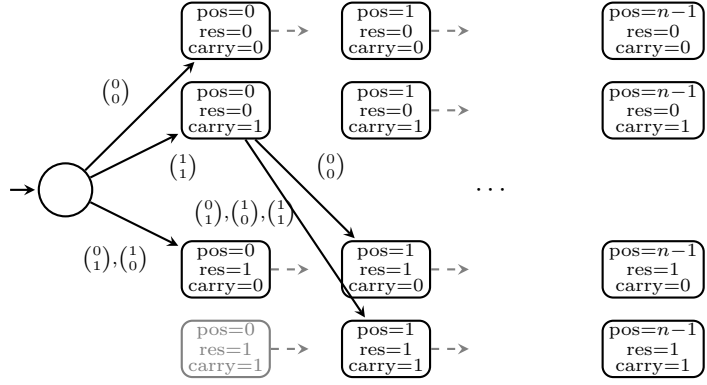

We first explain the idea of \cref{item:i3}.
The intention of the marker $\yes$ is to be placed before two numbers $x_j$ and $x_{j'}$ such that $x_j = x_{j'} = y_0$ and they enclose a bad $y_0$-pair, i.e., $x_i < y_0$ for $j < i < j'$.

In order to check whether $x_j = x_{j'} = y_0$ it suffices to have three gadgets similar as in \cref{ex:gadget} for $n$-bit sequences.
The first gadget stores $y_0$, we index its states with $S$ for start,
The second and third gadgets are entered when the first resp.\ second $\yes$ is seen, and thus store $x_j$ and $x_{j'}$, we index its states with $M_1$ resp.\ $M_2$ for marked numbers.
Validating whether $x_j = x_{j'} = y_0$ is done via the acceptance condition by a formula similar to the formula presented in \cref{ex:gadgetdelay}, which we present further below.
Our formula also enforces that there are two marked numbers.

Instead of checking whether $x_i < y_0$ for $j < i < j'$, the automaton checks for another condition which is expressible with a small enough automaton and formula.
Namely, whether $x_i + y_i = y_0$ for $j < i < j'$.
Note that this is only equivalent to $x_i \le y_0$. However, recall \cref{footnote:jpairs}.
The intention is that, during a play, Player~$O$ can make this condition true if, and only if, $x_j$ and $x_{j'}$ indeed enclose a bad $y_0$-pair.

In order to check whether $x_i + y_i = y_0$ for all $j < i < j'$, one gadget is used that computes the addition of two $n$-bit sequences.
Such a gadget is depicted in \cref{fig:gadget-add}. We index its states by $A$ for addition.
This gadget is entered for every $n$-bit sequence between the first two marked sequences.
Expressing that the automaton accepts if $x_i + y_i = y_0$ for all $j < i < j'$ is expressed via a formula which we give further below.

Now, we briefly explain the overall automaton structure.
We construct a deterministic automaton, thus, \cref{item:i1,item:i2,item:i3} must be handled in parallel.
Note that the gadgets we introduced for \cref{item:i3} can easily be extended to check whether either sequence is invalid (i.e., whether \cref{item:i1} or \cref{item:i2} holds) without introducing new states for the gadgets.
The automaton needs two additional sink states $q_{\mathit{acc}}$ resp.\ $q_{\mathit{rej}}$ that are entered when the input resp.\ output sequence is not a valid encoding.

After the $S$-gadget and before the $M_1$-gadget, and after the $M_2$-gadget, it suffices to check whether \cref{item:i1} or \cref{item:i2} holds and go to the corresponding sink state.
This can be done with $\bigo(n)$ many states.

All in all, $\aut_n$ is a sequence of six different gadgets which are all of size $\bigo(n)$, so $\aut_n$ is of size $\bigo(n)$:
\begin{enumerate}
    \item The first gadget, $S$, stores~$y_0$ the first number picked by Player~$O$.
    \item The second gadget checks that the format of the input and output is correct.
    \item The third gadget, $M_1$, is entered when the first~$\yes$ occurs, say at the beginning of the $j$-th block. It stores $x_j$.
    \item The fourth gadget, $A$, is entered immediately after the $j$-th block and only left when the second~$\yes$ occurs, say at the beginning of the $j'$-th block. It checks that $x_i + y_i = y_0$ holds for all $j < i < j'$.
    \item The fifth gadget, $M_2$, is entered when the second~$\yes$ occurs. It stores $x_{j'}$, the number marked by the second $\yes$.
    \item The last gadget again checks that the format of the input and output is correct.
\end{enumerate}
Note that the gadgets~$S$, $M_1$, $M_2$, and $A$ also check the format of the input and output is correct. 
As explained earlier, this does not increase their size. 

Finally, we are ready to present the acceptance condition $\varphi_n$ of the automaton $\aut_n$.
We state its intended meaning in words first:
\begin{itemize}
  \item the rejecting sink is not seen, and
  \item the accepting sink is seen, or 
  \item for every position $k \in [0,n-1]$ and every bit value $b \in \bool$: 
  \begin{itemize}
    \item If $y_0(k) = b$, then $x_j(k) = b$ and $x_{j'}(k) = b$, i.e., the two marked numbers are equal to the first one in the output.
    \item If $y_0(k) = b$, then the $k$-th bit of $x_i + y_i $ is not $1-b$ for every $j < i < j'$, i.e., all pairs of numbers between the marked ones add up to the first one in the output.
  \end{itemize}
\end{itemize}

We define $\varphi_n = \neg q_{\mathit{rej}} \wedge (q_{\mathit{acc}} \vee \psi_n)$ with
\[
\psi_n =   
    \bigwedge_{k \in [0,n-1]} \bigwedge_{b \in \bool} \zero{k}{b} \rightarrow \left[\one{k}{b} \wedge \two{k}{b} \wedge \bigwedge_{b' \in \bool}\neg \add{k}{1-b}{b'}\right].
\]
Note that $\varphi_n$ is of size $\bigo(n)$.

Now, we show that 
\begin{itemize}
  \item Player~$O$ wins $\delaygame{L_n}$ for some constant delay function $f$, but
  \item Player~$I$ wins $\delaygamep{L_n}$ for every delay function $g$ with $\sum_{i = 0}^{n-1} g(i) \leq 2^{2^n}$.
\end{itemize}

We start by giving a constant delay function $f$ such that Player~$O$ has a winning strategy in $\delaygame{L(\aut_n)}$.
According to \cref{prop:jpair}, the sequence~$x_1\cdots x_{2^{2^n}}$ of numbers is guaranteed to contain a bad $j$-pair for some $j \in [0,2^n~-~1]$.
Hence, we let $f(0) = m$, where $m$ is some number large enough to ensure that $x_1,\ldots,x_{2^{2^n}}$ is included in the lookahead (assuming that Player~$I$ plays a valid input encoding).
If Player~$I$ plays a valid input encoding, Player~$O$'s strategy is to output some $y_0$ such that $x_1\cdots x_{2^{2^n}}$ contains a bad $y_0$-pair.
She marks the pair correctly with $\yes$, and can ensure that she produces output $y_i$ such that $x_i + y_i = y_0$ in between the marked numbers. At other positions, she plays arbitrary numbers, and no other marks.

We argue that $\varphi_n$ is satisfied.
The state $q_{\mathit{rej}}$ is never reached, as
if Player~$I$ does not play a valid input encoding, the state $q_{\mathit{acc}}$ is reached independently of Player~$O$'s moves and the formula is trivially satisfied.
Assume that Player~$I$ plays a valid input encoding, and Player~$O$ behaves as described.
Clearly, $\zero{k}{b} \rightarrow (\one{k}{b} \wedge \two{k}{b})$ for all $k \in [0,n-1]$ and $b \in \bool$.
Since for all numbers $y_i$ between the marked numbers $x_i + y_i = y_0$ holds, we also have that if $\zero{k}{b}$ has been seen, then the $k$-th bit of $x_i + y_i$ is equal to $b$.
Hence, only the states $\add{k}{b}{0}$ and $\add{k}{b}{1}$ can be seen.
This implies that $\neg \add{k}{1-b}{0} \wedge \neg \add{k}{1-b}{1}$ is true.

For the other direction, assume that $g$ is a delay function such that $\sum_{i = 0}^{n-1} g(i) \leq 2^{2^n}$.
We show that Player~$I$ wins $\delaygamep{L_n}$.
According to \cref{prop:jpair}, there exists a word $w$ over $[0,2^{n}-1]$ with $|w|=2^{2^n}-1$ that contains no bad $j$-pair for every $j \in [0,2^{n}-1]$.
Player~$I$'s strategy is to play a valid input encoding such that $x_i$ is the $n$-bit sequence that encodes~$w(i)$ for all $0 \leq i \leq 2^{2^n}-1$. 
After round $n-1$, Player~$O$ has finished $y_0$.
Furthermore, the condition $\sum_{i = 0}^{n-1} g(i) \leq 2^{2^n}$ ensures that Player~$I$ has not yet begun to spell the $n$-bit sequence that encodes $x_{2^{2^n}}$.
Hence, since $y_0$ is known to Player~$I$, he can pick a number $x \neq y_0$ and continues to produce a valid input encoding such that every $x_i$ encodes $x$ for all $i \geq 2^{2^n}$.

We show that Player~$I$ wins playing as described.
Player~$O$'s goal is to satisfy the formula~$\varphi_n$.
Clearly, she loses if her moves do not produce a valid output encoding (as $q_{\mathit{rej}}$ would be reached), so we assume that she produces a valid output encoding.
The state $q_{\mathit{acc}}$ is not reached, as it is only reached if the input sequence is invalid, hence, the other part of the formula must be satisfied.
Note that the implication $\zero{k}{b} \rightarrow (\one{k}{b} \wedge \two{k}{b})$ for all $k \in [0,n-1]$ and $b \in \bool$ can only be satisfied if Player~$O$ marks two numbers with $\yes$, because the $S$-gadget stores $y_0$ making the lefthand-side of the implication true.
However, since Player~$I$'s moves ensured that the input sequence does not contain a bad $y_0$-pair, the righthand-side of the implication 
\[
  \bigwedge_{k \in [0,n-1]} \bigwedge_{b \in \bool} \zero{k}{b} \rightarrow \left[\one{k}{b} \wedge \two{k}{b} \wedge \bigwedge_{b' \in \bool}\neg \add{k}{1-b}{b'}\right]
\]
is not satisfied.
We have shown that Player~$I$ wins $\delaygamep{L(\aut_n)}$. 
\end{proof}

\paragraph{Complexity.}

Next, we settle the complexity of solving delay games with \formula-\wDMA winning conditions by proving a \twoexp lower bound.
Our construction is a generalization of the \exptime lower bound for solving delay games with winning conditions given by deterministic parity automata~\cite{DBLP:journals/lmcs/KleinZ14}.

Intuitively, Player~$I$ produces a sequence of configurations of an alternating exponential-space Turing machine. 
Each of these is of exponential length, i.e., the cells of each configuration can be addressed with polynomially many bits.
Now, Player~$I$ picks the successor configuration of universal configurations while Player~$O$ picks the successor configuration of existential configurations (by picking a transition to apply).
Relying on the lookahead, Player~$O$ has always access to the full current configuration picked by Player~$I$ before she has to pick a transition to apply.
Dually, to account for the lookahead, Player~$I$ is allowed to copy a configuration to fill the lookahead while he waits for Player~$O$ to pick a transition.
Hence, two successive configurations played by Player~$I$ should either be equal or the second one should be a successor configuration of the first one. 

To force Player~$I$ to faithfully simulate the Turing machine, i.e., to only copy configurations or to pick a valid successor configuration (in particular the one determined by Player~$O$ in case of existential configurations), Player~$O$ can mark cells for the automaton to check for correctness: This is sufficient as an error by Player~$I$ manifests itself in a single wrongly updated or copied cell. 
Using the addresses of the cells and the markers used by Player~$O$, a small automaton can check whether there is an error or not.
Finally, Player~$O$ is able to correctly apply the markers, as she has enough lookahead to always observe the next two configurations, enough to spot errors introduced by Player~$I$.

Altogether, we obtain a delay game simulating an alternating exponential-space Turing machine, and \aexpspace $=$ \twoexp yields the desired result.

\begin{theorem}\label{thm:detmullerComplete}
  Solving delay games with \formula-\wDMA winning conditions is \twoexp-complete.
\end{theorem}

\begin{proof}
  The upper bound was shown in \cref{thm:upperbound}, we show the corresponding lower bound.

  We give a reduction from the word problem for alternating exponential space Turing machines.
  Since $\text{\twoexp} = \text{\aexpspace}$ \cite{ch1981st}, we obtain the desired result.

  Let $\TM = (Q = Q_\exists \uplus Q_\forall,\Sigma,q_\initmark,\Delta,q_A,q_R)$ be an alternating exponential space Turing machine, where $\Delta \subseteq Q \times \Sigma \times \Sigma \times Q \times \set{-1,0,1}$.
  Furthermore, let $p$ be a polynomial such that $2^{p}$ bounds the space consumption of $\TM$ and let $w \in \Sigma^*$ be an input word.
  We fix $n = p(|w|)$.
  Wlog., we assume that the accepting state $q_A$ and the rejecting state $q_R$ are equipped with a self-loop.
  Furthermore, wlog., we assume that either $q_A$ or $q_R$ is reached in every run of $\TM$ on $w$.
  An alternating exponential space Turing machine that does not have this property can be turned into an equivalent TM $\mathcal M'$ that satisfies the property as follows:
  The number of configurations that can be reached on $w$ is bounded (doubly-exponentially in~$n$).
  Hence, $\TM$ can be equipped with an exponentially-sized counter in order to detect that a configuration repetition has occurred (when the counter has surpassed the maximum) and reject if so.
  Then, $\mathcal M'$ is used as the starting point for the reduction instead of $\TM$.

  We construct an \formula-\wDMA $\aut$ of polynomial size in $(|\Delta| + n)$ such that $\TM$ accepts~$w$ if, and only if, Player~$O$ wins $\delaygame{L(\aut)}$ for some delay function~$f$.

  The idea is that, in the delay game, the players build a run of $\TM$ on $w$ in the form of a configuration sequence.
  Player~$I$ controls the universal states and Player~$O$ the existential ones.
  Furthermore, Player~$I$ spells all configurations with his moves.
  The configurations are of exponential size in $n$, so each cell of $\TM$'s tape is addressed by an $n$-bit counter.
  In between configurations there is a delimiter, either $C$ or $N$, to denote if the next configuration is a \emph{copy} of the previous one or a \emph{new} one.
  The need to copy configurations arises since, in the delay game, Player~$O$ is behind with her moves, so Player~$I$ needs to wait for Player~$O$'s pick if the current configuration is existential.

  We go into details.
  The automaton we construct uses the product alphabet 
  \[
  \left(\set{0,1,C,N,\text{\textvisiblespace}} \cup Q \cup \Sigma\right) \times \left(\set{\ok,\nok,\no} \cup \Delta \right).
  \]
  Recall that $Q,\Delta$, and $\Sigma$ are components of the Turing machine $\TM$ and are used to encode configurations and to pick transitions to apply.
 We first describe what we call valid input resp.\ output formats:
  \begin{itemize}
    \item An input sequence $\alpha$ is valid if it is of the form 
    \[
        \left(\left(C+N\right)\big[\left(0+1\right)^n\left(\Sigma + Q + \text{\textvisiblespace}\right)\big]^+\right)^\omega,
    \]
    to be interpreted as a sequence of configurations $(C+N)c_0(C+N)c_1\cdots$ where each $c_i$ is a configuration delimited by $C$ or $N$.
    A configuration is encoded by a sequence $\big[(0+1)^n(\Sigma + Q + \text{\textvisiblespace})\big]^+$, where each $n$-bit sequence encodes an address $a \in [0,2^n-1]$, the least significant bit is the rightmost one, and the letter after the address encoding is either the tape cell content (that is, a letter $\sigma \in \Sigma$, or a blank (denoted as \textvisiblespace)) or a state $q \in Q$.
    Note that the format does not imply that a sequence in $\big[(0+1)^n(\Sigma + Q + \text{\textvisiblespace})\big]^+$ encodes a configuration of length~$2^n$. 
    This will be later enforced by the rules of the game. 

    \item An output sequence $\beta$ is valid format if it is of the form 
    \[
      \left(\Delta\big[(\ok + \nok + \no)^{n+1}\big]^+\right)^\omega,
    \]
    and if for every position $i$, $\alpha(i) \in \set{C,N}$ if, and only if, $\beta(i) \in \Delta$, i.e., the letters~$\tau \in \Delta$ occur at the same positions as the configuration delimiters $C$ and $N$.
    We give the intended interpretation of this format, which is explained in more detail below:
    The markers $\ok$ resp.\ $\nok$ and $\no$ are intended to indicate the absence of resp.\ the existence of an error.
    If $\nok$ marks a bit in an address block, it should indicate an error in updating the address block at this bit, if $\no$ marks the first bit of an address block, then it should indicate that there is a copy or update mistake in the successive configuration which manifests itself at the cell with the marked address.
  \end{itemize}

The behavior of the automaton is as follows:
\begin{enumerate}
  \item\label[item]{item:j1} If there exists some $i$ such that $\alpha(0)\alpha(1)\cdots\alpha(i)$ cannot be completed to have a valid format, but $\beta(0)\beta(1)\cdots\beta(i-1)$ can be completed to have a valid format, the automaton accepts.
  \item\label[item]{item:j2} If a configuration encoding $c$ contains not exactly one state or the first address block is not zero, the automaton accepts.
  \item\label[item]{item:j3} If $c_0$ does not encode the initial configuration on $w$, it accepts.
  \item\label[item]{item:j4} If there exists some $i$ such that $\beta(0)\beta(1)\cdots\beta(i)$ cannot be completed to have a valid format, but $\alpha(0)\alpha(1)\cdots\alpha(i)$ can be completed to have a valid format, the automaton rejects.
\end{enumerate} 
  We refer to \cref{item:j1,item:j2,item:j3} as simple input errors, and to \cref{item:j4} as simple output error.
  These conditions can be easily checked using the automaton structure with $\bigo(|\Delta| + n)$ states.
  We use designated sink states $q_{\mathit{acc}}$ and~$q_{\mathit{rej}}$.
\begin{enumerate}\setcounter{enumi}{4}
  \item\label[item]{item:j6} If there are no simple errors, and there is a position marked by $\nok$ in some address block $a_i$ (say $a_i(k)$), the automaton checks if $a_{i+1}(k)$ (assuming there is no $a_i(k')$ for some $k' < k$ which is also marked by $\nok$) is wrongly updated in the next address block $a_{i+1}$.
  We call this an \emph{address update error}.
  If the address is indeed wrongly updated, then the automaton accepts, otherwise it rejects. We explain further below how to achieve this behavior.
  \item\label[item]{item:j7} If there are no simple errors, and there is an address block $a$ whose first position is marked by $\no$ the following is checked:
  \begin{itemize}
    \item The next configuration must contain an address block $a'$ whose first position is marked by $\no$, otherwise the automaton rejects.
    (This is easily checkable via the automaton structure.)
    \item If $a \neq a'$, the automaton rejects.
    \item If $a = a'$, check whether the tape cell addressed by $a$ contains an error:
    \begin{itemize}
      \item If the delimiter between the configurations is $C$: The cell content is not correctly copied. We call this \emph{copy error}.
      \item If the delimiter between the configurations is $N$: The cell content is not correctly updated according to the picked transitions. We call this \emph{update error}.
    \end{itemize} 
    If an error is claimed correctly, the automaton accepts, otherwise it rejects.
  \end{itemize}
  We explain further below how to achieve this behavior.
  \item\label[item]{item:j8} Lastly, if none of the above cases occurs, if there is a configuration that contains $q_A$, then the automaton accepts, if there is a configuration that contains $q_R$, then the automaton rejects.
  Therefore, when a rejecting resp.\ accepting configuration has been seen, the automaton goes into designated new sink states $q_R$ resp.~$q_A$.
\end{enumerate}

Now, we explain how to handle \cref{item:j6}, i.e., the detection of an address update error.
When a bit is marked with $\nok$, its value is stored in the state space.
Using a modulo counter, it is easy to find the same bit in the next address block.
It is not difficult to see that the value of a bit only flips if all remaining bits (that is, all bits to the right) are one, because then adding one to the address counter causes an overflow of all less significant bits (which are denoted to the right of the marked bit).
Hence, an automaton can check with $\bigo(n)$ states whether a marked address bit has been updated correctly.

Regarding \cref{item:j7}, checking whether the two marked (with $\no$) $n$-bit sequences are the same can be done by using two gadgets and a formula as explained in \cref{ex:gadget,ex:gadgetdelay}.
Such a formula is of size~$\bigo(n)$.

In order to check whether there is a copy error between successive configurations it suffices to remember whether the delimiter between them is $C$ and the cell content of the marked cell.

Checking whether there is an update error between successive configurations $c$ and $c'$ is more involved:
\begin{itemize}
  \item The delimiter between them has to be $N$.
  \item A window of three cell contents around the cell whose address is marked must be remembered.
  \item If $c$ is a universal configuration, it has to be verified (using the stored three-cell window) whether the update in $c'$ is achievable by applying a transition from $\Delta$.
  (It is not specified in the input sequence which transition should be applied.)
  \item If $c$ is an existential configuration, the transition to be applied is specified in the output sequence.
  However, the transition to be applied has been given (possibly much) earlier:
  When the delimiter before a configuration is $N$, the configuration is new.
  If this configuration is existential, in the output sequence, the letter after the configuration specifies the transition to be applied.
  \item Transitions that are given after copied configurations as well as universal configurations are irrelevant. 
  \item Thus, the automaton stores a specified transition only if it occurs after a new existential configuration.
  This is the transition to be applied to obtain the next new configuration.
\end{itemize}
An automaton needs a polynomial number of states (in $(|\Delta| + n)$) to store the required pieces of information and verify that there is an update error.

Since we are constructing a \wDMA, all of \cref{item:j1,item:j2,item:j3,item:j4,item:j6,item:j7,item:j8} need to be  checked simultaneously.
Therefore we build a product automaton which we denote by $\aut$ with one special feature, namely, the introduced sink states $q_{acc},q_{rej},q_A,q_R$ are global sink states (as opposed to product states).

This product automaton has a number of states polynomial in $(|\Delta| + n)$ since it is constructed from gadgets whose number of states are polynomial in $(|\Delta| + n)$.
As mentioned before, regarding \cref{item:j7}, the formula needed to verify that two $n$-bit addresses are the same is of size $\bigo(n)$.
Since $\aut$ is a product automaton, the formula describing that two address are equal needs to be adjusted to incorporate the product structure which causes a polynomial in $(|\Delta| + n)$ blow-up of the formula.
We briefly explain how to take a product structure into account by a small example.
Let $\mathfrak C_1,\mathfrak C_2,\mathfrak C_3$ be some gadgets with state sets $Q_{1},Q_{2},Q_{3}$, respectively, and the product automaton is of the form $\mathfrak C_1 \times \mathfrak C_2 \times \mathfrak C_3$.
Say a formula only refers to states in $\mathfrak C_1$ which contains the state $p$.
Then the updated formula must replace all occurrences of $p$ with $\bigvee_{q \in Q_{2}} \bigvee_{r \in Q_{3}} (p,q,r)$.

All in all, the final formula $\varphi$ looks as follows:
\[
  \varphi = \neg (q_\mathit{rej} \vee {q_R}) \wedge ( q_\mathit{acc} \vee {q_A} \vee \varphi_{\mathit{error?}}),
\]
where $\varphi_{\mathit{error?}}$ describes that either the gadgets to store the $n$-bit address sequence have not been entered (meaning no copy/update error needs to be verified), or that they store the same $n$-bit address sequence (the claimed copy/update error is verified via the automaton structure).
The formula $\varphi_{\mathit{error?}}$ is adjusted to the product automaton structure.

Recall that $n = p(|w|)$, hence, $\varphi$ is of polynomial size in $(|\Delta| + p(|w|))$.
All in all, $\aut$ is of polynomial size in $(|\Delta| + p(|w|))$.

It is left to show that $\TM$ accepts $w$ if, and only if, Player~$O$ wins $\delaygame{L(\aut)}$ for some delay function~$f$.

To begin with, assume $\TM$ accepts $w$.
Let $f$ be a constant delay function such that $f(0)~=~m$, where $m$ is chosen such that Player~$O$ has enough lookahead to see a full configuration before she has to start with its output.
To be more concrete, $m = 2^n (1 + n) + 1$ suffices to capture a whole configuration including the address blocks and the delimiter.

We show that Player~$O$ wins $\delaygame{L(\aut)}$.
We assume that either player does not make simple errors (recall that this means playing in the wrong format).
Firstly, assume that Player~$I$ either does not update an address correctly, or introduces a copy or update error.
The lookahead is large enough to correctly claim such an error.
The formula is then satisfied.
Secondly, assume that Player~$I$ does not introduce such errors.
Since $\TM$ accepts $w$, Player~$O$ can pick existential transitions such that if Player~$I$ applies the picked transitions an accepting configuration will be reached.
Note that, in order to pick an existential transition such that an accepting configuration will be reached, it is only necessary to have the knowledge of the current configuration.
This fact can be easily seen when visualizing the run tree of $\TM$ on $w$.
If an accepting configuration is seen, the formula $\varphi$ is clearly satisfied.
Player~$I$ can prevent that an accepting configuration is seen by always copying the current configuration.
But then, the formula $\varphi$ is also satisfied, because the sub-formula $\varphi_{\mathit{error?}}$ is satisfied if no copy/update error is claimed or a copy/update error is correctly claimed.
Player~$O$ does not claim such an error, as we assumed that Player~$I$ does not introduce such an error.

For the other direction, assume $\TM$ rejects $w$.
We show that Player~$I$ wins $\delaygamep{L(\aut)}$ for every delay function~$g$.
Again, we assume that either player does not make simple format errors.
No matter the existential transitions that Player~$O$ choses, Player~$I$ can always pick universal transitions such that eventually a rejecting configuration is reached.
Hence, we assume that Player~$I$ advances the run by giving new configurations according to his and Player~$O$'s choices.

We go into more detail.
Whenever Player~$I$ has provided a new existential configuration, he has to wait for Player~$O$ to provide a transition to be applied since she is behind with her moves.
While Player~$O$ has not provided her transition choice, Player~$I$ copies the current configuration until the choice of Player~$O$ is known.
Then, Player~$I$ produces a new configuration obtained by applying Player~$O$'s choice.
Whenever Player~$I$ has provided a new universal configuration there is no need to copy the configuration, a new configuration can be provided directly after.

Furthermore, we assume that Player~$I$ does not introduce address update errors or configuration copy/update errors as this is not beneficial for him.
Player~$O$ loses a play where a rejecting configuration is reached.
Hence, her only chance is to claim an error. 
Since there are no errors introduced by Player~$I$, this approach is not fruitful as this consequently falsifies $\varphi$.
Player~$O$ cannot not win a play, i.e., Player~$I$ wins $\delaygamep{L(\aut)}$.
\end{proof}

\section{Lower Bounds for Non-Deterministic Weak Muller Delay Games}
\label{sec:nondetmuller}

In this section, we show lower bounds for delay games with non-deterministic weak Muller automata winning conditions.

\paragraph{Lookahead.}

We show a triply-exponential lower bound on the necessary lookahead for Player~$O$ to win, which yields a tight bound in combination with \cref{thm:upperboundLookahead}.

To this end, we again implement the bad $j$-pair game described in \cref{ex:jpair}, this time with numbers in the range~$[0,2^{2^n}-1]$. As argued before, this allows Player~$O$ to win with triply-exponential lookahead, but not with less.
Thus, we show that winning conditions specified by non-deterministic automata allow to implement the bad $j$-pair game with exponentially larger numbers than deterministic automata.
This increase requires some additional mechanism to allow the automaton to check the existence of a bad $j$-pair, which we describe below.

The binary encoding of each number in $[0,2^{2^n}-1]$ is of exponential length, i.e., each bit can be addressed with $n$ additional address bits. 
Player~$I$ produces a sequence of such encodings of numbers and Player~$O$ picks a number~$i_0$ in her first move, claiming that the sequence picked by Player~$I$ contains a bad $i_0$-pair.
To verify this claim, Player~$O$ is required to repeat $i_0$ ad infinitum.

Thus, both players can now cheat during their moves, Player~$I$ by not updating the addresses correctly and Player~$O$ by not copying $i_0$.
Thus, both players need to use markers to claim such errors, unlike in the analogous result for deterministic automata. 

However, due to the lookahead, Player~$I$ cannot mark the position where Player~$O$ has wrongly copied $i_0$. 
However, such an error manifests itself in a single bit and Player~$I$ is able to mark its address at a later position.
As there are exponentially many addresses, we rely on non-determinism to guess for each address whether it is equal to the marked one, i.e., it needs to be checked for a copy error, or whether it is not equal to the marked address.
It is this exponential succinctness, that leads to the increase in required lookahead from doubly-exponential to triply-exponential when allowing non-determinism.

\begin{theorem}\label{thm:nondetmullerLowerboundLookahead}
  For every $n \in \nats_{\geq 1}$, there exists a language $L_n$ recognized by an \formula-\wNMA $\aut_n$ of size $\bigo(n^3)$ such that
  \begin{itemize}
    \item Player~$O$ wins $\delaygame{L_n}$ for some constant delay function $f$, but
    \item Player~$I$ wins $\delaygamep{L_n}$ for every delay function $g$ with $\sum_{i = 0}^{(n+2) \cdot 2^n} g(i) \leq 2^{2^{2^n}}$.
  \end{itemize}
\end{theorem}

\begin{proof}
  Pick some $n \in \nats_{\geq 1}.$
  Our goal is to give an~\formula-\wNMA~$\aut_n$ of  polynomial size in $n$ such that Player~$O$ needs triply-exponential lookahead in $n$ to win the corresponding delay game.
  Recall \cref{ex:jpair,prop:jpair}.
  We are constructing a bad $j$-pair automaton over numbers ranging between $0$ and $2^{2^{n}}-1$ which enforces the need of a word composed of at least $2^{2^{2^n}}$ numbers to be guaranteed the existence of a bad $j$-pair.

  To build an automaton of polynomial size in $n$, we encode a number \mbox{$x\in R = [0,2^{2^{n}}-1]$} as a word over $\bool$ as follows:
  To represent a number $x \in R$ in binary one needs $2^n$ bits, we address each bit by an $n$-bit sequence.

  We go into more detail.
  A sequence of the form $a_0 b_0 \cdots a_{2^n-1} b_{2^n-1}$ such that $a_i \in \bool^n$ and $b_i \in \bool$ for $0 \leq i \leq 2^n - 1$ is called a \emph{superblock}.
  A block $a_i \in \bool^n$ is interpreted as an \emph{address block} whose least significant bit is the rightmost one.
  The bit sequence $b_0\cdots b_{2^n - 1}$ is interpreted as a number $x \in R$, the least significant bit is the rightmost one.
  We refer to these bits as \emph{special bits}.

  In order to explain the idea, assume that in a play, Player~$I$ plays an input sequence that encodes a sequence of numbers $x_0,x_1,\ldots \in R$.
  We give a construction such that Player~$O$'s moves must produce an output sequence that encodes a sequence of numbers $y_0,y_1,\ldots \in R$ such that $x_1x_2\cdots$ contains a bad $y_0$-pair and it holds that $y_0 = y_1 = y_2 = \cdots$.
  In words, we enforce that Player~$O$ has to copy her first number ad infinitum.
  The technique to enforce Player~$O$ to copy her first number ad infinitum is as follows:
  When Player~$I$ realizes that she has not faithfully copied her number, then there exists an address $a$ such that there are two superblocks whose values of the special bits addressed by $a$ differ.
  Player~$I$ marks an address block that encodes $a$, and the winning condition (i.e., the automaton) enforces that all special bits which are addressed by $a$ must have the same value.
  This is then violated, so Player~$O$ loses if she does not copy her choice of $y_0$ ad infinitum.

  Furthermore, the automaton we construct verifies that the input sequence indeed contains a bad $y_0$-pair.
  Since it is enforced that $y_0 = y_1 = y_2 = \cdots$, verifying whether a bad $y_0$-pair exists is doable with a small enough automaton which we explain further below.
  
  We go into details.
  The automaton we construct uses the product alphabet~$\set{0,1,\yes,\#} \times \set{0,1,\#,\ok,\nok}$.
  We first describe what we call valid input resp.\ output format:
  \begin{itemize}
    \item An input sequence $\alpha$ has a valid format if it is of the form
    \[
      \left((\yes + \#)(0+1)^{n+1}\right)^\omega,
      \]
      where $\yes$ occurs exactly once.
    \item An output sequence $\beta$ has a valid format if it is of the form
    \[
      \big(\#(\ok + \nok)^{n}(0+1)\big)^\omega.
      \]
      The intention is that $\ok$ resp.\ $\nok$ are used to mark the absence resp.\ the presence of errors in updating the address blocks.
  \end{itemize}

  We now describe how the automaton $\aut_n$ behaves, then how to achieve this behavior.
  \begin{enumerate}
    \item\label[item]{item:k1} If there exists some $i$ such that $\alpha(0)\alpha(1)\cdots\alpha(i)$ cannot be completed to have a valid format, but $\beta(0)\beta(1)\cdots\beta(i-1)$ can be completed to have a valid format, the automaton accepts.
    \item\label[item]{item:k2} If there exists some $i$ such that $\beta(0)\beta(1)\cdots\beta(i)$ cannot be completed to have a valid format, but $\alpha(0)\alpha(1)\cdots\alpha(i)$ can be completed to have a valid format, the automaton rejects.
    
    \item\label[item]{item:k3} If the very first address is not zero, it accepts.
    \item\label[item]{item:k4} If the $k$-th bit of an address block is marked with $\nok$ (and no earlier $\nok$ occurs), then it is verified if the $k$-th bit of the next address block is correctly updated. If the update is faulty, the automaton accepts, if the update is correct, the automaton rejects.
    \item\label[item]{item:k5} If $\yes$ marks an address $a$, then the special bits picked by Player~$O$ that are addressed with $a$ either have all value zero or all have value one.
    If this condition is violated, the automaton rejects.
    \item\label[item]{item:k6} If none of the conditions before have lead to acceptance or rejection, the automaton verifies whether $x_1x_2\cdots $ contains a bad $y_0$-pair.
    If the answer is yes, it accepts, otherwise it rejects.
  \end{enumerate}

  We now describe how to achieve the desired behavior.
  Note that \cref{item:k1,item:k2,item:k3} can be easily checked in parallel via the automaton structure of an automaton with $\bigo(n)$ states.
  We use designated sink states $q_{\mathit{acc}}$ resp.\ $q_{\mathit{rej}}$ that are reached when the automaton accepts resp.\ rejects.

  Note that \cref{item:k4} can be handled as \cref{item:j6} introduced in the proof of \cref{thm:detmullerComplete} where we re-use the sink states $q_{\mathit{acc}},q_{\mathit{rej}}$.
  Again, $\bigo(n)$ states are used.

  It is left to handle \cref{item:k5,item:k6}.
  Regarding \cref{item:k5}, we need one gadget that stores the $n$-bit sequence that has been marked by $\yes$.
  We refer to this sequence as address $a$.
  Such a gadget uses $\bigo(n)$ states and was first introduced (and used) in \cref{ex:gadget,ex:gadgetdelay}.
  We index its states with $A$ for address.

  The automaton needs to verify that all special bits that are addressed by $a$ have the same value.
  Therefore, the automaton behaves as follows:
  \begin{itemize}
    \item Before a new address block begins, the automaton guesses whether this address will be equal to $a$.
    \item If the automaton guesses that the addresses are equal:
    \begin{itemize}
      \item It must be verified that the guess  is  right.
      If the guess  is  wrong, the automaton rejects.
      
      \item In order to do the verification we need one additional gadget that stores an $n$-bit sequence such as in \cref{ex:gadget}.
      We index its states by $A_{=}$ for equal address.
      The automaton enters this gadget every time it guesses that the address will be equal to $a$.
      Further below, we give a formula that is satisfied if, and only if, the guess  is  always right. 
      
      \item Additionally, the automaton must store the values (as a set) of the special bits picked by Player~$O$ that occur as the next letter.
    \end{itemize}
    \item If the automaton guesses that the addresses are not equal:
    \begin{itemize}
      \item It must be verified that the guess  is  right.
      If the guess  is  wrong, the automaton rejects.
      \item To verify, guess some $i \in [0,n-1]$ and check that the value of the $i$-th address bit is different from the $i$-th bit of the address $a$ (recall, which is marked by $\yes$, and stored in gadget $A$).
      \item In \cref{fig:gadget-guess} we give a gadget that implements the guessing and storing of a bit.
      We index its states by $G$ for guessing.
      The automaton enters this gadget every time it guesses that the address will be not equal to $a$.
      \item Further below, we give a formula that is satisfied if, and only if, always a bit could be found such that the values at this bit in the inspected address and in address $a$ differ.
    \end{itemize}
  \end{itemize}
  Such an automaton needs $\bigo(n)$ states.

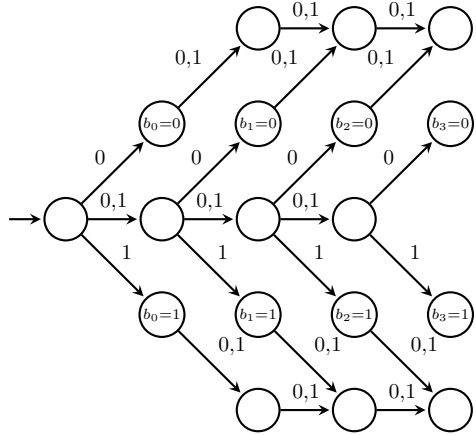
\begin{figure}[t]
  \hspace*{\fill}
  \begin{tikzpicture}[thick]
    \tikzstyle{every state}+=[inner sep=1pt, minimum size=2em, node distance=4.5em];

    \tikzstyle{every node}+=[scale=0.8]

    \node[state, initial] (s) {};

    \node[state, right of = s] (s0) {};
    \node[state, right of = s0] (s1) {};
    \node[state, right of = s1] (s2) {};

    \node[state, right of = s, yshift = 4.5em] (00) {\nodelbl{0}{0}};
    \node[state, right of = 00] (10) {\nodelbl{1}{0}};
    \node[state, right of = 10] (20) {\nodelbl{2}{0}};
    \node[state, right of = 20] (30) {\nodelbl{3}{0}};

    \node[state, above of = 10] (a1) {};
    \node[state, above of = 20] (a2) {};
    \node[state, above of = 30] (a3) {};

    \node[state, right of = s, yshift = -4.5em] (01) {\nodelbl{0}{1}};
    \node[state, right of = 01] (11) {\nodelbl{1}{1}};
    \node[state, right of = 11] (21) {\nodelbl{2}{1}};
    \node[state, right of = 21] (31) {\nodelbl{3}{1}};

    \node[state, below of = 11] (b1) {};
    \node[state, below of = 21] (b2) {};
    \node[state, below of = 31] (b3) {};

    \draw[->] (s) edge[auto] node {0} (00);
    \draw[->] (00) edge[auto] node {0,1} (a1);
    \draw[->] (10) edge[auto] node {0,1} (a2);
    \draw[->] (20) edge[auto] node {0,1} (a3);

    \draw[->] (s) edge[auto] node {1} (01);
    \draw[->] (01) edge[auto] node {0,1} (b1);
    \draw[->] (11) edge[auto] node {0,1} (b2);
    \draw[->] (21) edge[auto] node {0,1} (b3);

    \draw[->] (a1) edge[auto] node {0,1} (a2);
    \draw[->] (a2) edge[auto] node {0,1} (a3);

    \draw[->] (b1) edge[auto] node {0,1} (b2);
    \draw[->] (b2) edge[auto] node {0,1} (b3);

    \draw[->] (s0) edge[auto] node {0,1} (s1);
    \draw[->] (s1) edge[auto] node {0,1} (s2);

    \draw[->] (s) edge[auto] node {0,1} (s0);

    \draw[->] (s0) edge[auto] node {0} (10);
    \draw[->] (s1) edge[auto] node {0} (20);
    \draw[->] (s2) edge[auto] node {0} (30);

    \draw[->] (s0) edge[auto] node {1} (11);
    \draw[->] (s1) edge[auto] node {1} (21);
    \draw[->] (s2) edge[auto] node {1} (31);

  \end{tikzpicture}
  \hspace*{\fill}
  \caption{A gadget to (non-deterministically) store a single bit from a $4$-bit sequence.
  The top and bottom row of states serve the same purpose (namely, a bit has been stored) so they could be merged, but it is easier to draw the gadget in this way.
  This gadget is easily generalized to an $n$-bit sequence.}
  \label{fig:gadget-guess}
\end{figure}

Lastly, regarding \cref{item:k6}, the automaton guesses before the beginning of some superblock that its encoded number $x_j$ is such that $x_j = y_0$. 
Whether a position is the beginning of a superblock can be guessed and verified by the automaton, as all its special bits are zero.
Since \cref{item:k5} ensures that $y_j = y_0$ and $y_j$ occurs below $x_j$ it is easy to verify that the guess  is  correct (if the guess is wrong, the automaton rejects).
Then, for the following numbers $x_{j+1}, x_{j+2},\ldots$ it has to verify that $x_{j+1} < y_{j+1}, x_{j+2} < y_{j+2}, \ldots$ until some $x_{j'}$ is seen such that $x_{j'} = y_{j'}$.
To verify that some number $x$ is smaller than $y$, first, recall that $x$ and $y$ are encoded by their sequence of special bits such that the least significant bit is the rightmost one.
If $x < y$, there exists a special bit at some position $i$ such that all previous special bits had the same value for $x$ and $y$, the special bit at position $i$ has value zero in $x$ and value one in $y$.
The values of the special bits to the right play no role.
If the verification fails, the automaton rejects by going to some new sink $q_{\mathit
{verif}-\mathit{fail}}$.
If some $x_{j'} = y_{j'}$ has been seen, the automaton goes into a new accepting sink $q_{\mathit{verif}-\mathit{ok}}$.
Such an automaton needs a constant number of states, as the positions of special bits can be easily identified in a valid output sequence.

Recall that we construct a non-deterministic automaton, hence, the desired automaton~$\aut_n$ is given as 
\[
  (\aut_{\mathit{format}} \times \aut_{\mathit{address}}) \cup (\aut_{\mathit{format}} \times \aut_{\mathit{copy}} \times \aut_{\mathit{pair}})
\]
where $\aut_{\mathit{format}}$ handles \cref{item:k1,item:k2,item:k3}, $\aut_{\mathit{address}}$ handles \cref{item:k4}, $\aut_{\mathit{copy}}$ handles \cref{item:k5}, and $\aut_{\mathit{pair}}$ handles \cref{item:k6}.
The automaton has $\bigo(n^3)$ states.

The acceptance condition is given by the formula $\varphi_n$ defined as 
\[
 \neg q_{\mathit{rej}} \wedge \neg q_{\mathit
 {verif}-\mathit{fail}} \wedge \left( q_{\mathit{acc}} \vee (q_{\mathit
 {verif}-\mathit{ok}} \wedge \varphi_{\mathit{copy}})\right),
\]
where $\varphi_{\mathit{copy}}$ is defined as
\[
  \begin{array}{ll}
   & \bigwedge_{k \in [0,n-1]}  \bigwedge_{b \in \bool}  \addequal{k}{b} \rightarrow \addi{k}{b}\\
   \wedge & \bigwedge_{k \in [0,n-1]} \bigwedge_{b \in \bool} \addi{k}{b} \rightarrow \neg \addguess{k}{b}\\
   \wedge & \varphi_{\mathit{same}-\mathit{value}}.
  \end{array}
\]
The first line of $\varphi_{\mathit{copy}}$ is satisfied if, and only if, the automaton whenever it has guessed \myquot{current address equals marked address} is right.
If a such guess is wrong there is some $k \in [0,n-1]$ and $b \in \bool$ such that $\addi{k}{b}$ and $\addequal{k}{1-b}$ have been visited, falsifying the formula.

The second line of $\varphi_{\mathit{copy}}$ is satisfied if, and only if, the automaton whenever it has guessed \myquot{current address is not equal to marked address} is right.
Since the current address and the marked address differ in the value of at least one bit, the automaton can always guess to store this different bit in the gadget $G$.
Hence, the second line of the formula is satisfiable in this way.

The formula $\varphi_{\mathit{same}-\mathit{value}}$ expresses that the relevant special bits either all have value zero, or all have value one.

Note that the formula $\varphi_n$ needs to be adjusted to talk about the (product-like) structure of $\aut_n$ (which can be done as explained in the proof of \cref{thm:detmullerComplete}), hence, its final size is~$\bigo(n^3)$.

It is left to prove that 
\begin{itemize}
\item Player~$O$ wins $\delaygame{L(\aut_n)}$ for some constant delay function $f$, but
\item Player~$I$ wins $\delaygamep{L(\aut_n)}$ for every delay function $g$ with $\sum_{i = 0}^{(n+2)\cdot {2^n}} g(i) \leq 2^{2^{2^n}}$.
\end{itemize}

We start by giving a constant delay function $f$ such that Player~$O$ has a winning strategy in $\delaygame{L(\aut_n)}$.
According to \cref{prop:jpair}, the sequence~$x_1\cdots x_{2^{2^{2^n}}}$ of numbers is guaranteed to contain a bad $j$-pair for some $j \in [0,2^{2^n}~-~1]$.
Hence, we let $f(0) = m$, where $m$ is some number large enough to ensure that $x_1,\ldots,x_{2^{2^{2^n}}}$ is included in the lookahead (assuming that Player~$I$ plays valid superblocks with either $\#$ or $\yes$ (used once) in front of an address block).
We call playing valid superblocks with either $\#$ or $\yes$ in front of an address block playing valid for short.

If Player~$I$ plays valid, Player~$O$'s strategy is to output some $y_0$ such that $x_1\cdots x_{2^{2^{2^n}}}$ contains a bad $y_0$-pair.
The number $y_0$ is repeated ad infinitum.

If Player~$I$ does not play valid, Player~$I$ either violates the format, or he introduces an error in updating the address counter.
In the former case, $q_{\mathit{acc}}$ is reached and
the formula~$\varphi_n$ is trivially satisfied.
In the latter case, Player~$O$ has enough lookahead to realize this and mark the error. 
Then, also $q_{\mathit{acc}}$ is  reached.
The state $q_{\mathit{rec}}$ is never reached as she does not violate the desired format.

The sub-formula $\varphi_{\mathit{copy}}$ is also satisfied as she faithfully copies her first number $y_0$ ad infinitum.
It is left to explain why $q_{\mathit
{verif}-\mathit{fail}}$ is not reached and $q_{\mathit
{verif}-\mathit{ok}}$ is reached.
Since the sequence $x_1x_2\cdots$ contains a bad $y_0$-pair, it possible for the automaton to pick a run that correctly verifies this, hence, only $q_{\mathit
{verif}-\mathit{ok}}$ is reached.
Thus, we have shown that the formula $\varphi_n$ is satisfied when Player~$O$ plays according to her strategy.

We turn to the other direction.
Let $g$ be a delay function with $\sum_{i = 0}^{(n+2)\cdot {2^n}} g(i) \leq 2^{2^{2^n}}$.
We show that Player~$I$ wins $\delaygamep{L(\aut_n)}$.

According to \cref{prop:jpair}, there exists a word $w$ over $[0,2^{2^{n}}-1]$ with $|w|=2^{2^{2^n}}-1$ that contains no bad $j$-pair for every $j \in [0,2^{2^{n}}-1]$.
Player~$I$'s strategy is to play a valid input sequence such that $x_i$ is the bit sequence that encodes~$w(i)$ for all $0 \leq i \leq 2^{2^{2^n}}-1$. 

Recall, a superblock is of length $(n+2) \cdot {2^n}$ as we have $2^n$ address blocks of length $n$, $2^n$ special bits, and, additionally, $2^n$ address \myquot{markers} ($\#$ signifying unmarked, $\yes$ signifying marked).
After round $(n+2)\cdot {2^n}$, Player~$O$ has finished $y_0$.
Furthermore, the condition $\sum_{i = 0}^{(n+2)\cdot {2^n}} g(i) \leq 2^{2^n}$ ensures that Player~$I$ has not yet begun to spell the superblock that encodes $x_{2^{2^{2^n}}}$.
Hence, since $y_0$ is known to Player~$I$, he can pick a number $x \neq y_0$ and continue to produce a valid input sequence such that every $x_i$ encodes $x$ for all $i \geq 2^{2^n}$.

We show that Player~$I$ wins playing as described.
Player~$O$'s goal is to satisfy the formula~$\varphi_n$.
Clearly, she loses if her moves have the wrong format (as $q_{\mathit{rej}}$ would be reached), so we assume that she produces an output sequence that has the right format.
The state $q_{\mathit{acc}}$ is not reached, as it is only reached if the input sequence is invalid, hence, $\neg q_{\mathit
{verif}-\mathit{fail}} \wedge \left( q_{\mathit
{verif}-\mathit{ok}} \wedge \varphi_{\mathit{copy}}\right)$ must be satisfied.
We can assume that Player~$O$ faithfully copies $y_0$ to satisfy $\varphi_{\mathit{copy}}$.
But, since the input sequence does not contain a bad $y_0$-pair, no matter which run $\aut_n$ takes, the verification of a bad $y_0$-pair is never successful.
Hence, $q_{\mathit
{verif}-\mathit{fail}}$ is reached, falsifying $\varphi_n$.
Player~$I$ wins.
\end{proof}

\paragraph{Complexity.}

Finally, we turn to complexity showing a \threeexp lower bound on solving delay games with \formula-\wNMA winning conditions.
To this end, we generalize the analogous construction for \formula-\wDMA presented in the proof of \cref{thm:detmullerComplete}: 
Here, the players simulate a doubly-exponential space Turing machine.
Thus, the use of non-deterministic automata for the winning condition allows us to add an additional exponential blowup in the space consumption of the Turing machine.
Again, this requires some additional mechanisms, e.g., both players have to check their opponent's moves for correctness.
We describe these changes in comparison to the construction for deterministic automata below.

As before, Player~$I$ is in charge of producing the sequence of configurations (with repetitions to fill the lookahead) and of picking successors of universal configurations while Player~$O$ picks successors for existential configurations (by picking a transition for Player~$I$ to apply).
As configurations are now of doubly-exponential length, we address their cells by exponentially long addresses. 
Hence, the bits of these addresses have to be addressed using linearly-sized addresses as well, i.e., there are two levels of addresses.
Player~$O$ has again markers to force Player~$I$ to faithfully simulate the Turing machine while Player~$I$ also has markers to ensure that Player~$O$ correctly marks errors.
Altogether, we obtain a delay game that simulates an alternating doubly-exponential space Turing machine, and \atwoexpspace $=$ \threeexp yields the desired lower bound.

\begin{theorem}\label{thm:nondetmullerComplete}
  Solving delay games with \formula-\wNMA winning conditions is \threeexp-complete.
\end{theorem}

\begin{proof}
  The upper bound was shown in \cref{thm:upperbound}, we show the corresponding lower bound.

  We give a reduction from the word problem for alternating doubly-exponential space Turing machines.
  Since $\text{\threeexp} = \text{\atwoexpspace}$ \cite{ch1981st}, we obtain the desired result.

  Let $\TM = (Q = Q_\exists \uplus Q_\forall,\Sigma,q_\initmark,\Delta,q_A,q_R)$ be an alternating doubly-exponential space Turing machine, where $\Delta \subseteq Q \times \Sigma \times \Sigma \times Q \times \set{-1,0,1}$.
  Furthermore, let $p$ be a polynomial such that $2^{2^p}$ bounds the space consumption of $\TM$ and let $w \in \Sigma^*$ be an input word.
  We fix $n = p(|w|)$.
  Wlog., we assume that the accepting state $q_A$ and the rejecting state $q_R$ are equipped with a self-loop.
  Furthermore, wlog., we assume that either $q_A$ or $q_R$ is reached in every run of $\TM$ on $w$.
  We have explained in the proof of \cref{thm:detmullerComplete} that such an assumption can be made.
  
  We construct an \formula-\wNMA $\aut$ of polynomial size in $(|\Delta| + n)$ such that $\TM$ accepts~$w$ if, and only if, Player~$O$ wins $\delaygame{L(\aut)}$ for some delay function~$f$.

  As in the proof of \cref{thm:detmullerComplete}, the idea is that, in the delay game, the players build a run of $\TM$ on $w$ in the form of a configuration sequence.
  Player~$I$ controls the universal states and Player~$O$ the existential ones.
  Furthermore, Player~$I$ spells all configurations with his moves.

  The configurations are of doubly-exponential size in $n$, so each cell of $\TM$'s tape is addressed by a superblock as introduced in the proof of \cref{thm:nondetmullerLowerboundLookahead}.

  In between configurations there is a delimiter, either $C$ or $N$, to denote if the next configuration is a \emph{copy} of the previous one or a \emph{new} one.
  The need to copy configurations arises since, in the delay game, Player~$O$ is behind with her moves, so Player~$I$ needs to wait for Player~$O$'s pick if the current configuration is existential.

  We go into details.
  The automaton we construct uses the product alphabet 
  \[\left(\set{0,1,C,N,\text{\textvisiblespace},\yes,\#} \cup Q \cup \Sigma\right) \times \left(\set{0,1,\#,\no,\noo,\ok,\nok} \cup \Delta\right).\]
  Recall that $Q,\Delta$, and $\Sigma$ are components of the Turing machine $\TM$ and are used to encode configurations and to pick transitions to apply.
  We first describe what we call valid input resp.\ output formats:
  \begin{itemize}
    \item An input sequence $\alpha$ is valid if it is of the form 
    \[
        \left((C+N)\big[S(\Sigma + Q + \text{\textvisiblespace})\big]^+\right)^\omega,
    \]
    where $S$ stands for superblock and is defined as
    \[
      S = \left((\yes + \#)(0 +  1)^{n+1}\right)^+.
    \]
    \begin{itemize}
      \item The sequence $\alpha$ is to be interpreted as a sequence of configurations 
      \[(C~+~N)c_0(C~+~N)c_1\cdots\]
      where each $c_i$ is a configuration delimited by $C$ or $N$.
      Note that the correct encoding of such configurations is enforced by the rules of the game.
      \item In a configuration encoding, a superblock defines the number of a cell, the letter after the superblock is either the tape content (that is, a letter $\sigma \in \Sigma$, or a blank (denoted as \textvisiblespace)) or a state $q \in Q$. 
      \item Recall that a superblock $S$ is to be interpreted as first a marker~$\yes$ or $\#$ to mark (the position before) an address block inside a superblock, then an $n$-bit address block encoding an address~$a$, followed by a special bit $b$.
      Then again, a marker, an address block, a special bit, a marker, and so on.
      \item The sequence of special bits of a superblock defines the number of the cell.
      \item The marker $\yes$ can be used once in the input sequence to mark (the position before) an address inside a superblock.
      The purpose is explained further below.
      \item The letter $\#$ is to be interpreted as unmarked.
    \end{itemize}    

    \item An output sequence $\beta$ is valid if it is of the form 
    \[
      \left(\Delta\big[ (\no + \noo + \#)(\ok + \nok)^{n} (0+1)\big]^+\right)^\omega,
    \]
    and if for every position $i$, $\alpha(i) \in \set{C,N}$ if, and only if, $\beta(i) \in \Delta$, i.e., the letters~$\tau \in \Delta$ occur at the same positions as the configuration delimiters $C$ and $N$.
    We give the intended interpretation of this format, which is detailed below.
    \begin{itemize}
      \item The markers $\ok$ resp.\ $\nok$ are intended to indicate the absence of resp.\ the existence of an error an address block.
      \item If $\nok$ marks a bit in an address block, it should indicate an error in updating the address block at this bit.
      \item If $\no$ marks (the position before) an address block, say the block encodes address~$a$, then it should indicate that there is an error in updating the special bit with address~$a$ in the next superblock.
      \item Finally, if $\noo$ marks (the position before) an address block, then it should indicate that there is a copy or update mistake in the successive configuration.
    \end{itemize}    
  \end{itemize}

The behavior of the automaton is as follows.
We first describe the occurrence of simple format errors that lead to acceptance resp.\ rejection:
\begin{enumerate}
  \item\label[item]{item:l1} If there exists some $i$ such that $\alpha(0)\alpha(1)\cdots\alpha(i)$ cannot be completed to have a valid format, but $\beta(0)\beta(1)\cdots\beta(i-1)$ can be completed to have a valid format, the automaton accepts.
  \item\label[item]{item:l2} If a configuration encoding $c$ contains not exactly one state or the first address block in a superblock is not zero, or the first superblock of a configuration does not encode the number zero, the automaton accepts.
  \item\label[item]{item:l3} If $c_0$ does not encode the initial configuration on $w$, it accepts.
  \item\label[item]{item:l4} If there exists some $i$ such that $\beta(0)\beta(1)\cdots\beta(i)$ cannot be completed to have a valid format, but $\alpha(0)\alpha(1)\cdots\alpha(i)$ can be completed to have a valid format, the automaton rejects.
\end{enumerate} 
\cref{item:l1,item:l2,item:l3,item:l4} can be handled similar to \cref{item:j1,item:j2,item:j3,item:j4} in the proof of \cref{thm:detmullerComplete}.
 One needs automata of polynomial size in $(|\Delta| + n)$ to do so.

 We now describe different kinds of (more involved) errors and how they are handled:
\begin{enumerate}
  \setcounter{enumi}{4}
  \item\label[item]{item:l5} If the $k$-th bit of an address block is marked with $\nok$, then it is verified if the $k$-th bit of the next address block is correctly updated. If the update is faulty, the automaton accepts, if the update is correct, the automaton rejects.
 \end{enumerate}
 This item can be handled as \cref{item:j6} in the proof of \cref{thm:detmullerComplete} with $\bigo(n)$ states.
 \begin{enumerate}
  \setcounter{enumi}{5}
  \item\label[item]{item:l6} If $\no$ marks an address $a$ in a superblock, then the automaton checks whether the special bit that is addressed by $a$ in the next superblock has been correctly updated.
  If an error has been correctly claimed, the automaton accepts, otherwise it rejects.
 \end{enumerate}
To achieve this behavior, a gadget as introduced in \cref{ex:gadget} is needed that stores the address marked by $\no$.
To determine how the value of the special bit addressed by the marked address changes it suffices to check whether all special bits (in the same superblock) to the right are one.
If at least one zero is seen, the value of the special bit does not flip.
Otherwise it does.
In the next superblock, the automaton has to find the special bit that has the same address.
It has to guess when it has advanced to this address block and verify that the guess is correct.
Therefore, we need another gadget as introduced in \cref{ex:gadget} to store the guessed bit and compare via a formula that the two stored addresses are equal.
The gadgets need $\bigo(n)$ many states, checking whether the special bit is correctly updated can be done with a constant number of states.
 \begin{enumerate}
  \setcounter{enumi}{6}
  \item\label[item]{item:l7} If the $\noo$ marker occurs in a superblock $S$, then the automaton checks whether there is a \emph{copy} or \emph{update error} between successive configurations that manifests in the cell whose number is encoded by $S$.
  If an error had been correctly claimed, the automaton accepts, otherwise it rejects.
 \end{enumerate}
 We have detailed in the proof of \cref{thm:detmullerComplete} (regarding \cref{item:j7}) how a copy or update error is checked when the relevant pieces of information are known (which requires $\bigo(|\Delta| + n)$ states).
 The difficulty lies in finding the same cell in the current and in the successor configuration.
 Therefore, the automaton guesses at the beginning of a superblock, say $S$, whether it has some address that is marked with $\noo$.
 (A wrong guess leads to rejection.)
 If the superblock $S$ contains a $\noo$, then the automaton enforces that, from then on, the output must copy the special bits of $S$ ad infinitum.
 This condition makes it easy to find the cell addressed by $S$ in the next configuration, because it suffices to find the cell addressed by a superblock such that the special bits in the input and output sequence are equal.

 We have described in detail in the proof of \cref{thm:nondetmullerLowerboundLookahead} (see \cref{item:k5}) how to enforce that a bit sequence is faithfully copied.
 We use the same technique here, which needs $\bigo(n)$ states.
 To recap, in case of an error, the input sequence marks an address of a special bit with $\yes$.
 This address is where a copy error has manifested.

 The last condition is concerned with whether $w$ is accepted by $\TM$:
 \begin{enumerate}
  \setcounter{enumi}{7}
   \item\label[item]{item:l8} If the format is valid and no error occurred so far: If a configuration with $q_A$ resp.\ $q_R$ is seen, the automaton accepts resp.\ rejects.
 \end{enumerate}
 This can be handled as \cref{item:j8} in the proof of \cref{thm:detmullerComplete} with a constant number of states.

All in all, the automaton $\aut$ needs a polynomial number of states (in $|\Delta| + n$).

 Finally, we are ready to describe the acceptance condition of $\aut$ as a formula $\varphi$ defined as 
\[
  \neg q_{\mathit{rej}} \wedge \neg q_{R} \wedge \left( q_{\mathit{acc}} \vee q_A \vee \varphi_{\no} \vee \varphi_{\noo}\right).
\]
 The formula $\varphi_{\no}$ is satisfied if the gadgets to store two addresses needed to find an update error regarding a special bit have not been entered (no special bit error is claimed) or the stored addresses are equal.
 
 The formula $\varphi_{\noo}$ is satisfied if the gadget to check faithful copying is either not entered or the gadget is entered and indicates that copying is successful.
 There are two reasons why the copy gadget is not entered.
 The marker $\noo$ is not used in the output, i.e., no configuration copy or update error is claimed.
 Or, $\noo$ is used, but the marker $\yes$ is not used in the input, i.e., no error in copying the special bits is claimed.

The size of $\varphi$ is polynomial in $(|\Delta| + n)$.

It is left to show that $\TM$ accepts $w$ if, and only if, Player~$O$ wins $\delaygame{L(\aut)}$ for some delay function~$f$.
To begin with, assume $\TM$ accepts $w$.
Let $f$ be a constant delay function such that $f(0) = m$, where $m$ is chosen such that Player~$O$ has enough lookahead to see a full configuration before she has to start with her output.
This means $m$ is triply-exponential in $n$.
We show that Player~$O$ wins $\delaygame{L(\aut)}$.
We assume that either player does not make simple errors (recall that means playing in the wrong format).
Firstly, assume that Player~$I$ either does not update an address block correctly, does not update a special bit in a superblock correctly, or introduces a configuration copy or update error.
The lookahead is large enough for Player~$O$ to correctly claim such an error.
The formula is then satisfied.
Secondly, assume that Player~$I$ does not introduce such errors.
Since $\TM$ accepts $w$, Player~$O$ can pick existential transitions such that if Player~$I$ applies the picked transitions an accepting configuration will be reached.
If an accepting configuration is seen, the formula~$\varphi$ is clearly satisfied.
Player~$I$ can prevent that an accepting configuration is seen by always copying the current configuration.
But then, the formula $\varphi$ is also satisfied, because the sub-formulas $\varphi_{\no}$ and $\varphi_{\noo}$ are satisfied if no error of any kind is claimed.
Player~$O$ does need to claim an error, as we assumed that Player~$I$ does not introduce an error.

For the other direction, assume $\TM$ rejects $w$.
We show that Player~$I$ wins $\delaygamep{L(\aut)}$ for every delay function~$g$.
Again, we assume that either player does not make simple format errors.
No matter the existential transitions that Player~$O$ chooses, Player~$I$ can always pick universal transitions such that eventually a rejecting configuration is reached.
Hence, we assume that Player~$I$ advances the run by giving new configurations according to his and Player~$O$'s choices.
A more detailed description what it means to advance a run is given in the proof of \cref{thm:detmullerComplete}. 

Furthermore, we assume that Player~$I$ does not introduce any kind of errors as this is not beneficial for him.
Player~$O$ loses a play where a rejecting configuration is reached.
Hence, her only chance is to claim an error. 
Since there are no errors introduced by Player~$I$, this approach is not fruitful as this consequently falsifies $\varphi$.
Player~$O$ cannot win a play, i.e., Player~$I$ wins $\delaygamep{L(\aut)}$.
\end{proof}

\section{Conclusion}
\label{sec:conclusion}

We took the first step towards investigating delay games with Muller conditions by showing doubly- and triply-exponential bounds for deterministic and non-deterministic delay games with weak Muller conditions.
More, specifically we have shown that solving such games is \twoexp-complete (for deterministic automata) and \threeexp-complete (for non-deterministic automata) and that doubly-exponential (for deterministic automata) and triply-exponential lookahead (for non-deterministic automata) is necessary and sufficient.
These results have to be compared to those for deterministic safety and parity automata, for which solving delay games is \exptime-complete and exponential lookahead is necessary and sufficient.
Similarly, for non-deterministic safety and parity automata solving delay games is \twoexp-complete and doubly-exponential lookahead is necessary and sufficient.
Thus, the succinctness of weak Muller conditions (encoded by formulas) yields an exponential increase in comparison to both safety and parity conditions.

There are two immediate directions for further research: Try to lift our results to (standard) Muller conditions and to strengthen the lower bounds by proving them for less succinct representations of $\curlyF$. 

Recall that our lower bounds for weak Muller conditions rely on gadgets comparing $n$-bit strings for some fixed $n$.
This is possible for a fixed number of comparisons by having one gadget for each string to be compared and then using the formula defining $\curlyF$ to implement the actual comparison.
A natural way to lift, say, the bad $j$-pair lower bound game is to just play this game infinitely often.
However, this requires to reuse the gadgets infinitely often, which then means one cannot compare strings from one iteration of the bad $j$-pair game, but compares strings from all iterations.

Let us stress again that all our upper bounds are independent of the representation of the acceptance condition. 
However, our lower bounds only hold for Emerson-Lei conditions, which is the most succinct representation.
Strengthening the lower bounds via less succinct representations most likely requires a new approach, as all lower bounds we have proven require the comparison of $n$-bit strings. 
This can be done with a small formula, as shown here, but not with the less succinct presentations considered by Hunter and Dawar for delay-free games~\cite{DBLP:conf/mfcs/HunterD05,HunterPhD}.
In particular, it is even open whether delay games with explicit Muller conditions (the least succinct representation) are harder than delay games with parity or safety conditions. 
Note that in the delay-free case, explicit Muller games can be solved in polynomial time~\cite{DBLP:conf/fsttcs/Horn08}.

\bibliographystyle{plain}
\bibliography{biblio}

\end{document}